\documentclass[aps,superscriptaddress,onecolumn,10pt,prx]{revtex4-2}
\PassOptionsToPackage{table}{xcolor}
\usepackage{silence}
\usepackage{amsmath,mathtools,amsthm,amssymb}
\usepackage{tabularx}
\usepackage{tabularray}
\usepackage{tikz}
\usepackage{tabularx}
\usepackage{tabularray}

\usepackage{graphicx}

\usepackage{color}
\usepackage{bbold}
\usepackage{enumitem}
\usepackage{centernot}
\usepackage{complexity}
\usepackage{physics}
\usepackage{comment}
\usepackage{array}
\usepackage{graphics}
\usepackage{wrapfig}
\usepackage{fontsize}
\graphicspath{{figures/}}
\usepackage{biolinum}
\usepackage{bbm}
\usepackage{dsfont}
\usepackage{mathrsfs}
\usepackage{mathdots}
\usepackage{enumitem}
\usepackage[table]{xcolor}
\usepackage{booktabs}
\usepackage{amssymb}
\usepackage{xcolor}
\newcommand{\id}{\mathrm{Id}}

\newcommand{\be}{\begin{equation}}
\newcommand{\ee}{\end{equation}}
\newcommand{\ba}{\begin{eqnarray}}
\newcommand{\ea}{\end{eqnarray}}
\newcommand{\ignore}[1]{}

\definecolor{lightred}{RGB}{243,229,231}
\definecolor{lightgreen}{RGB}{241,255,239}
\definecolor{lightblue}{RGB}{232,240,244}
\definecolor{RoyalBlue}{RGB}{65,105,225}
\definecolor{ForestGreen}{RGB}{34,139,34}   % RGB for ForestGreen
\definecolor{Maroon}{RGB}{135,0,0}
\definecolor{myrefcolor}{rgb}{0.067,0.5,0.5}
\definecolor{myurlcolor}{rgb}{0.1,0,0.9}

\usepackage[
    breaklinks,
    pdftex,
    colorlinks=true,
    linkcolor=myrefcolor,
    citecolor=myrefcolor,
    urlcolor=myrefcolor
]{hyperref}
\usepackage[capitalize]{cleveref}

\usepackage{hyperref}

\newtheorem{theorem}{Theorem}
\newtheorem{proposition}{Proposition}
\newtheorem{lemma}{Lemma}
\newtheorem{corollary}{Corollary}
\newtheorem{definition}{Definition}

\newtheorem{remark}{Remark}

\begin{document}

\title{Non-local Magic: closed-form solution and equivalence with magic of purification}

\author{Michele Viscardi}
\affiliation{Dipartimento di Fisica, Università degli Studi di Salerno, Via Giovanni Paolo II, 132, 84084 Fisciano (SA), Italy}

\author{Lorenzo Leone}

\affiliation{Dipartimento di Ingegneria Industriale, Università degli Studi di Salerno, Via Giovanni Paolo II, 132, 84084 Fisciano (SA), Italy}
\affiliation{Istituto Nazionale di Fisica Nucleare (INFN) Sezione di Napoli, Gruppo Collegato di Salerno, Italy}

\author{Alioscia Hamma}
\affiliation{Scuola Superiore Meridionale, Largo S. Marcellino 10, 80138 Napoli, Italy}
\affiliation{Istituto Nazionale di Fisica Nucleare (INFN) Sezione di Napoli, Italy}
\affiliation{Università degli Studi di Napoli Federico II , Dipartimento di Fisica Ettore Pancini, Italy}

\begin{abstract}

Non-local magic quantifies the non-stabilizerness beyond the local choice of basis. Its evaluation is a hard problem requiring optimization over local unitaries.  Here, we show that for the log-stabilizer fidelity this optimization admits an exact closed-form analytic solution depending only on the Schmidt spectrum. We then address the open problem of non-stabilizerness of mixed states by introducing the magic of purification and show that it naturally induces a resource theory whose free states are the states with zero stabilizer entropy.  For the log-stabilizer fidelity, the magic of purification admits a distance-based formulation in terms of the Uhlmann fidelity. We prove that non-local magic coincides with the minimum magic of purification along the unitary orbit of the reduced density operator. This framework introduces non-local magic as the relevant resource for non-trivial steering and the characterization of entangling and non-stabilizer resources in measurements. These results provide both an efficient analytical characterization and a mixed-state resource-theoretic interpretation of non-local magic.

\end{abstract}
\maketitle

\section{Introduction}

A long-standing goal of quantum information theory is to understand where quantum complexity comes from, a question closely tied to the origin of quantum computational advantage. Entanglement, the most familiar signature of quantum behavior, is the natural first candidate. Yet entanglement alone does not make a quantum state complex. Indeed, stabilizer operations are capable of generating even maximally entangled many-body states, while remaining efficiently simulable on a classical computer~\cite{AaronsonGottesman2004}.
Whatever makes a quantum state genuinely complex must therefore come from resources that lie outside the stabilizer framework. These resources are commonly referred to as \textit{non-stabilizerness}, or
\textit{magic}~\cite{BravyiKitaev2005,Veitch2014}.
This observation motivated the formulation of magic resource theories, in which stabilizer states and stabilizer operations are considered as free, and the amount of magic in a state can be quantified by a variety of monotones~\cite{LiuWinter2022,LeoneOlivieroHamma2022} (i.e. functions that are non-increasing under free operations). 
Magic was first studied for its role in fault-tolerant quantum computation~\cite{BravyiKitaev2005}, but has recently emerged as a useful probe of complexity in quantum many-body systems~\cite{LiuWinter2022,White2021conformal,Tarabunga2024criticalbehaviorsof,Viscardi_2026,esposito2026stabilizerentropytrustworthymixed}, complementing more conventional quantities based on entanglement.

Entanglement and magic capture fundamentally different properties of a quantum state. Entanglement is invariant under local changes of basis, whereas magic is defined with respect to a distinguished Pauli frame and can therefore be extensively modified by local unitaries. A product state, for instance, may possess a large amount of magic while containing no non-local correlations whatsoever.
These two resources meet through a structural property of quantum states. Across any bipartition, a pure stabilizer state is equivalent under local Clifford unitaries to a collection of $k$ Bell pairs shared between the two parts, together with unentangled stabilizer degrees of freedom~\cite{Fattal2004,PRXQuantum.6.020324}. The entanglement spectrum of pure stabilizer states is therefore always flat on $2^k$ values.
By contrast, a state with non-flat entanglement spectrum cannot be turned into a stabilizer state by any local change of basis. The part of magic that survives every such change of local basis is referred to as \emph{non-local magic}~\cite{Cao_2025, QianWang2025}. It is defined as the minimum value of a magic monotone over local unitaries acting on the two parts of the system. Non-local magic is therefore invariant under local changes of frame and resides only in the correlations between subsystems. From this perspective, it is the natural measure of non-stabilizerness in quantum many-body systems. Indeed, total magic can grow with the system volume even at vanishing correlation length, as for a product state written in a tilted local frame, and is therefore decoupled from the length scale that governs many-body physics. Non-local magic, instead, vanishes identically for product states and can only originate from correlations. For such reasons, a growing body of work in literature is exploring the relationships between different monotones and witnesses of non-local magic and quantum many-body phenomena~\cite{Korbany2025,QianWang2025, Catalano2025generalizedWstates}.

Non-local magic is closely tied to the anti-flatness of the entanglement spectrum~\cite{Tirrito2024,Cao_2025,jasser2026journeyflatlanddoesantiflatness}, and several recent works have adopted spectral quantities to estimate or characterize it in a variety of settings~\cite{Torre2026,franchini2026schmidtgaugenonlocalmagicrepresentation,cusumano2025nonstabilizernessviolationschshinequalities,iannotti2026nonlocalmagicresourcesfermionic,cao2024non,Cepollaro2025harvesting,odavic2023complexityoffrustration,Robin_2026,liu2026entanglementantiflatnessnonlocalnonstabilizerness,Ebner2026magicbarrier,wei2026entirelynonlocalquantummagic,chu2026magicphasephaseindependentstabilizer}. It has also been studied in holography~\cite{Cao_2025}, and has recently been experimentally measured on a superconducting quantum processor~ 
\cite{ahmad2025experimentaldemonstrationnonlocalmagic}. The relationship between non-local magic and the anti-flatness of the entanglement spectrum is the following: Clifford non-local operations populate the entanglement spectrum in a flat way, and locally injected non-stabilizerness bends it. Together, the two resources conspire as non-local magic to produce non-trivial entanglement spectra. This work clarifies this mechanism and tackles the problem of an exact solution for the computation of non local magic. Indeed, 
 the definition of non-local magic involves an optimization over two exponentially large unitary groups, making its direct evaluation prohibitive in general. Moreover, beyond its connection with the entanglement spectrum, a general resource-theoretic interpretation of this minimization remains to be understood. In particular, since the Schmidt spectrum of a bipartite pure state coincides with the spectrum of either reduced density operator, it is natural to ask whether non-local magic can be understood directly as a notion of magic for the corresponding mixed state.
%Despite this progress, three important problems remain unsolved.
%The first problem is the efficient evaluation of non-local magic. Its definition, indeed, involves an optimization over the unitary groups of the subsystems, making its direct evaluation prohibitive in general. 
%Secondly, since non-local magic is, by definition, a specific form of magic, it remains unclear how it relates to other resource theories of magic and, in particular, given its bipartite nature, how it is related to the magic of the reduced states.
Finally, although non-local magic has been connected to specific physical phenomena, it still lacks a general operational interpretation: it is not known which physical tasks require it, nor what it allows one to do that stabilizer resources cannot.
\begin{figure}
    \centering
    \includegraphics[width=0.9\linewidth]{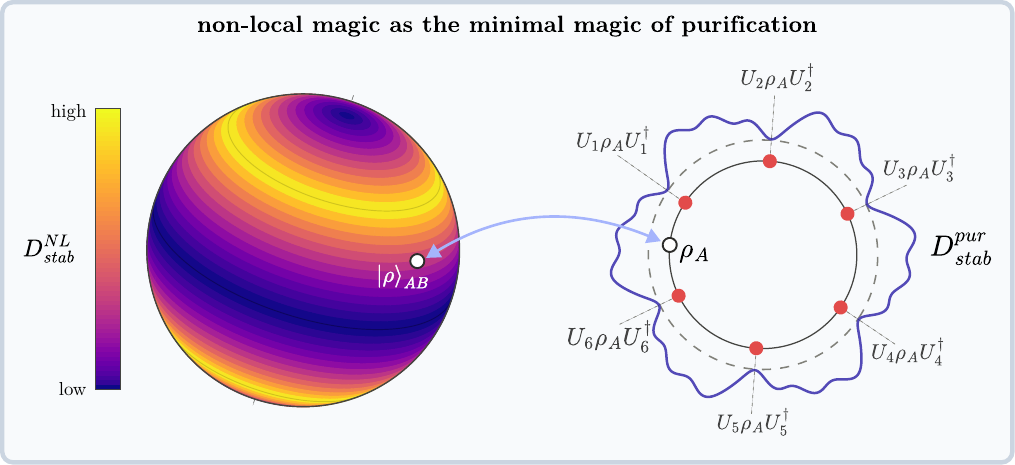}
    \caption{Pictorial representation of the relationship between the non-local log-stabilizer fidelity of a bipartite state and the minimal log-stabilizer fidelity of purification of the reduced state (Theorem~\ref{thm:orbitmin}). On the left, the Hilbert space of a bipartite system is represented as a sphere. A point on the sphere is a valid pure quantum state for such system. The colormap shows the values of the non-local log-stabilizer fidelity $D_{stab}^{NL}$. Local unitary orbits are represented as iso-valued circles on the sphere. After the partial trace of the subsystem $B$, a state $\ket{\rho}_{AB}$ is mapped to the reduced state $\rho_A$. On the left, we represent the unitary orbit of such a reduced state as a circle. The purple line represents the value of the log-stabilizer fidelity of purification $D_{stab}^{pur}$ for states $U\rho_AU^\dag$. The argmin of $D_{stab}^{pur}$ are shown as red dots on the unitary orbit of $\rho_A$. These states share the same minimal value of $D_{stab}^{pur}$, which is represented as a dashed circle connecting the respective values on the purple curve.}
    \label{fig:1}
\end{figure}

In this work, we address these open problems by focusing on the log-stabilizer fidelity, a particularly natural magic monotone based on the maximal fidelity with a pure stabilizer state~\cite{LiuWinter2022}. Evaluating the log-stabilizer fidelity $D_{\mathrm{stab}}$ is hard for a generic quantum state. Nevertheless, we show that its non-local counterpart $D_{\mathrm{stab}}^{\mathrm{NL}}$ admits a closed-form expression for every bipartite pure state $\ket{\psi}$, which depends only on its Schmidt coefficients $\theta_i^{\downarrow}$, arranged in decreasing order:
\begin{align}\label{eq1}
    D_{\mathrm{stab}}^{\mathrm{NL}}(\ket{\psi})=-\log\max_{k\le n}2^{-k}\left(\sum_{i=1}^{2^k}\theta_{i}^{\downarrow}\right)^2 \,,
\end{align}
The joint optimization over local unitaries and stabilizer states thus reduces to a maximization over a single integer $k$, the number of Bell pairs shared across the bipartition. The formula is a direct consequence of the quantization
of stabilizer entanglement. Since every stabilizer spectrum is flat on $2^k$ values, the best stabilizer approximation of $\ket{\psi}$ is, up to local unitaries, a collection of $k$ Bell pairs aligned with the $2^k$ largest Schmidt coefficients of $\ket{\psi}$. Non-local magic therefore measures how far the entanglement spectrum of $\ket{\psi}$ is from the closest flat spectrum of dyadic rank.
\cref{eq1} immediately reveals that the non-local log-stabilizer fidelity is much more tractable than its original variational definition might suggest. For an arbitrary state, its evaluation requires only $O(2^n)$ operations. At the same time, because the final expression depends exclusively on the Schmidt spectrum, the quantity is also experimentally accessible through entanglement-spectrum estimation.

\begin{figure}
    \centering
    \includegraphics[width=0.95\linewidth]{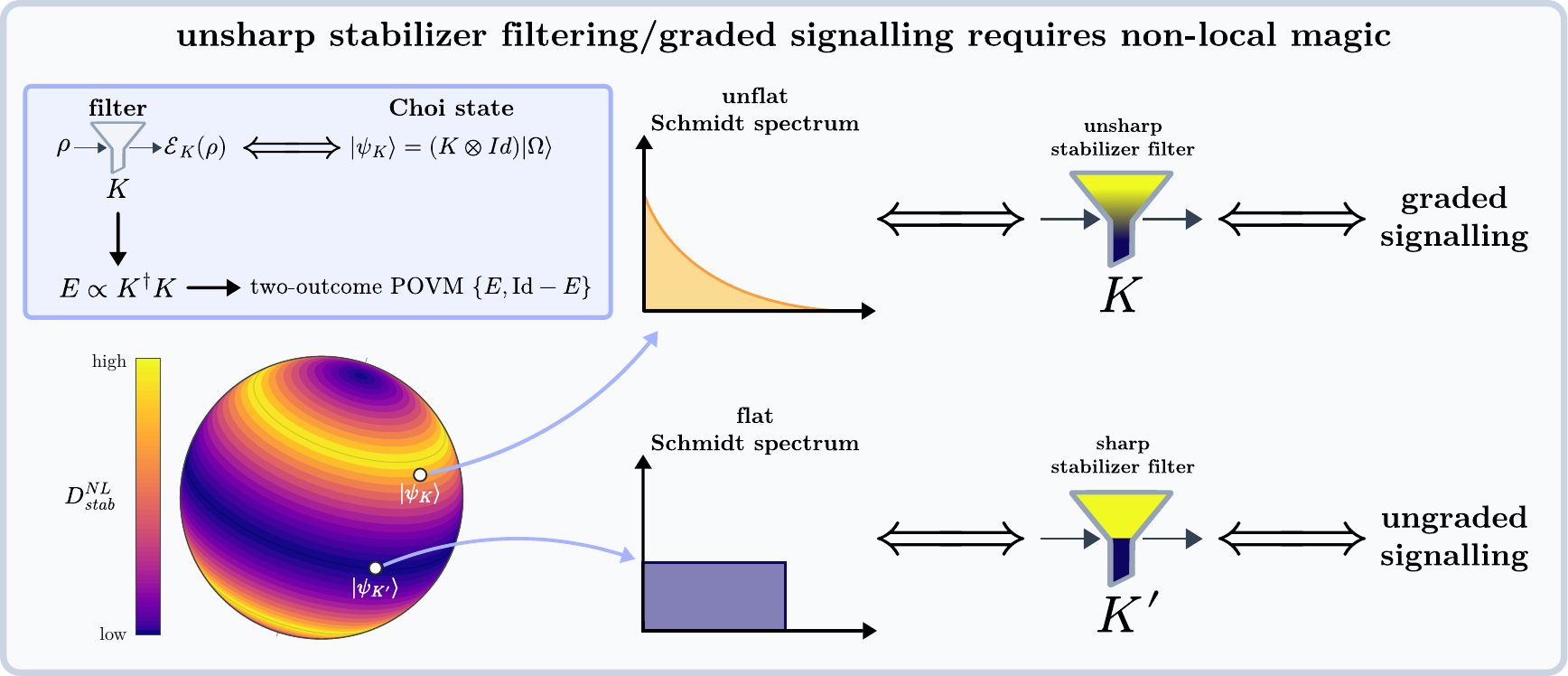}
    \caption{Pictorial representation of the relation between non-local magic, the unsharpness of a stabilizer filter and graded signalling. In the light-blue box, we consider a filter $K$, to which it is assigned, through the Choi--Jamio{\l}kowski correspondence, a pure bipartite quantum state $\ket{\psi_K}$. We represent the Hilbert space of such bipartite system as a sphere, and valid pure bipartite quantum states are represented as points on the sphere. As in Fig.~\ref{fig:1}, the colormap indicates the value of the non-local log-stabilizer fidelity $D_{stab}^{NL}$.  We also consider the effect $E$ of the filter $K$ and its corresponding two-outcome POVM. As proven in Corollary~\ref{cor:gradedfilters}, the flatness of the Schmidt spectrum of $\ket{\psi_K}$ witnesses if $K$ is a sharp stabilizer filter. Consequently, it can also witness if the associated effect $E$ is a stabilizer projector, and if the type of signalling is graded or ungraded.}
    \label{fig:2}
\end{figure}

The exact solution allows us to go beyond the spectral characterization and uncover a mixed-state interpretation of non-local magic. Inspired by the entanglement of purification~\cite{Terhal2002}, we introduce the \textit{magic of purification}: given a mixed state $\rho$, we define its purification magic as the minimum pure-state magic among all states that purify $\rho$. This construction provides a natural alternative to the usual convex-roof extension of pure-state magic monotones. The two constructions lead, in general, to different resource theories. The convex-roof extension, which we refer to as the \textit{magic of formation}, has as its free set the conventional convex hull of pure stabilizer states, denoted by $\mathrm{STAB1}$~\cite{LeoneBittel2024}. By contrast, the magic of purification naturally singles out the smaller set $\mathrm{STAB0}\subset\mathrm{STAB1}$ of mixed states admitting a stabilizer purification, equivalently normalized projectors onto stabilizer codes. Remarkably, this is characterized also as the set of states with vanishing stabilizer entropy \cite{LeoneOlivieroHamma2022}. We establish the resource-theoretic properties of this construction and show that, whenever the underlying pure-state monotone obeys strong monotonicity, the magic of formation is bounded from above by the magic of purification.

For the log-stabilizer fidelity, the purification construction takes an especially simple form. Using Uhlmann's theorem, we show that its magic of purification is exactly the negative logarithm of the maximum Uhlmann fidelity between $\rho$ and the set $\mathrm{STAB0}$. The purification optimization can thus be rewritten as the standard distance-to-the-free-set construction familiar from quantum resource theories~\cite{ChitambarGour2019}. This result clarifies the role of $\mathrm{STAB0}$ and gives the magic of purification a direct geometric interpretation in the space of mixed quantum states.

Most importantly, this framework provides an exact connection between mixed-state magic and non-local magic. For any density operator $\rho_A$ and any of its purifications $\ket{\rho_{AB}}$ across a bipartition $A|B$, we prove that the non-local log-stabilizer fidelity $D_{\mathrm{stab}}^{\mathrm{NL}}$ of the purification coincides with the minimum magic of purification $D_{\mathrm{stab}}^{\mathrm{pur}}$ attained along the complete unitary orbit of $\rho_A$:
\begin{align}
    D_{\mathrm{stab}}^{\mathrm{NL}}(\ket{\rho}_{AB})=\min_{U_A}D_{\mathrm{stab}}^{\mathrm{pur}}(U_A\rho_AU_A^{\dagger})\,.
\end{align}
Non-local magic can therefore be equivalently viewed in two ways: as the irreducible magic of a bipartite pure state under local changes of basis, or as the minimum purification magic compatible with the spectrum of its reduced density operator (see Fig.~\ref{fig:1} for a pictorial representation). We also show that every unitary orbit of a density operator contains a magic minimum belonging to $\mathrm{STAB1}$ and diagonal in a stabilizer basis (Remark~\ref{cor:stab1exists}). As a direct consequence, its purification is a minimizer of the log-stabilizer fidelity over local unitaries (Remark~\ref{remark:minimizerNL}). 

We then connect, through the Choi--Jamio{\l}kowski correspondence, non-local magic to the structure of quantum filters. 
We show that stabilizer Choi states correspond exactly to filters composed of Clifford unitaries and stabilizer projections, whose effects are projectors and therefore  accept or reject inputs in a all-or-nothing fashion. On the contrary, unsharp filters, which accepts different inputs  with intermediate probabilities, require a non-flat Choi spectrum, and hence nonzero non-local magic (see Fig.~\ref{fig:2}).

We then turn from single filters to quantum channels and
ask how much of the non-Clifford character of an operation can be injected \emph{locally}, by non-Clifford unitaries acting separately on a system and an ancilla coupled through a Clifford interaction. Projective measurements are the paradigmatic example: the coupling to the apparatus is a generalized CNOT, and all of the magic resides in the choice of measurement basis. We show that non-local magic is precisely the minimal amount of magic that must be injected locally for a Clifford interaction to produce a given entanglement spectrum, and that channels with magic injected on the system alone have flat Choi spectra.

Finally, we apply this framework to quantum steering, where Alice performs a measurement on a state which is shared with Bob, and, by choosing the measurement basis, can steer the set of states that Bob's system can be left in. 
We show that when the shared state and Alice's measurement basis are both stabilizer, then steering is sharp, i.e. Bob's possible states are mutually orthogonal with each other and equally probable. In this case, Bob's knowledge of Alice's outcome is all-or-nothing: by measuring his system he can determine with certainty which class of outcomes occurred, but he cannot gain any information that distinguishes outcomes within the same class.
Graded steering therefore requires magic, and it can come from two sources, either the shared state or Alice's measurement basis. If the shared state has non-local magic, every measurement of Alice steers Bob in a graded way. If the state does not possess non-local magic, steering can still be graded, but this is achieved entirely by the magic in Alice's choice of basis. The only two exceptions to the latter result are when the shared state is a product state or a maximally entangled state. In such cases, every measurement from Alice produces a sharp steering. 
We also quantify the magic budget of steering and show that non-local magic is its irreducible part: every measurement on a given state costs at least this much, and measuring in the Schmidt basis costs exactly this much.

Taken together, our results identify non-local magic as the resource that separates all-or-nothing from graded correlations, and place it at the intersection of three viewpoints: the entanglement spectrum of bipartite pure states, the resource theory of mixed-state magic, and the structure of quantum operations. The log-stabilizer fidelity provides an exactly solvable instance in which all three connections can be made explicit.

\section{Preliminaries}\label{sec:prelim}
We consider a system of $n$ qubits and denote with $d$ its dimension. The set of pure stabilizer states is denoted by $\mathrm{PSTAB}$. Its convex hull is denoted by $\mathrm{STAB1}$. The subset $\mathrm{STAB0}\in \mathrm{STAB1}$ consists of all the states that can be purified in stabilizer states. Equivalently $\mathrm{STAB0}=\{\Pi_S/2^{k}\}$ with $\Pi_S$ a rank-$2^{k}$ stabilizer code projector.

The free operations of the magic-state resource theory with free set $\mathrm{STAB1}$ are generated by combinations of the following elementary operations: Clifford unitaries, partial trace, measurements in the computational basis, composition with auxiliary qubits, and operations conditioned on measurement outcomes and classical randomness. By contrast, when the free set is taken to be $\mathrm{STAB0}$, conditioned operations may generate states outside the free set~\cite{Yashin_2025}. In the following, we denote by $\mathcal{S}_1$ and $\mathcal{S}_0$ the sets of stabilizer protocols associated with $\mathrm{STAB1}$ and $\mathrm{STAB0}$, respectively.

We now define stabilizer monotones for magic-state resource theory with the set of free states being $\mathrm{STAB1}$. We remark that analogous definitions are valid when the set of free states is chosen to be $\mathrm{STAB0}$.
\begin{definition}[Stabilizer monotones]
\leavevmode\par\label{def:stabilizermonotones}
\begin{itemize}
    \item A stabilizer monotone $\mathcal{M}$ is a real-valued function defined for all $n\in \mathbb{N}$ qubit systems, or collections thereof, such that: (i) $\mathcal{\mathcal{M}}(\rho)=0$ if and only $\rho\in\mathrm{STAB1}$; and (ii) $\mathcal{M}$ is nonincreasing under stabilizer protocols $\mathcal{E}$, i.e. $\mathcal{M}(\mathcal{E}(\rho))\le \mathcal{M}(\rho)$.
\item A pure-state stabilizer monotone instead satisfies $\mathcal{M}(\psi)\le \mathcal{M}(\phi)$ for any pair of pure states $\psi,\phi$ such that there exists a stabilizer protocol $\mathcal{E}$ with $\mathcal{E}(\phi)=\psi$.

\item A pure-state stabilizer monotone is said to be {\em strong} if, for every pure state $\ket{\psi}$ and every stabilizer protocol $\mathcal{E}$ producing the ensemble of pure states $\{(p_i,\ket{\phi_i})\}$, it holds that $\mathcal{M}(\psi)\ge \sum_i p_i\mathcal{M}(\phi_i)$.
\end{itemize}
\end{definition}

Let us define one of the main object of this work, which is the stabilizer fidelity and its associated magic monotone.
\begin{definition}[Stabilizer fidelity]\label{def:stabfid} Let $\ket{\psi}$ a pure quantum state, and $\mathrm{PSTAB}$ the set of pure stabilizer states. Then its stabilizer fidelity is defined as:
\begin{align}
    F_{\mathrm{stab}}(\ket{\psi})\coloneqq \max_{\sigma\in\mathrm{PSTAB}}|\langle\psi|\sigma\rangle|^2
\end{align}
while $D_{\mathrm{stab}}(\ket{\psi})\coloneqq-\log F_{\mathrm{stab}}(\ket{\psi})$ is referred to as the log-stabilizer fidelity~\cite{LiuWinter2022}. This measure is generalized to mixed quantum state by using the Uhlmann fidelity, which is $\|\sqrt{\rho}\sqrt{\sigma}\|_1^2$ for two arbitrary mixed quantum states $\rho$ and $\sigma$.
\end{definition}

Below, we also introduce the stabilizer R\'enyi entropy~\cite{LeoneOlivieroHamma2022} (SRE), which has become a standard tool for quantifying non-stabilzierness. Remarkably, its extension to mixed states can be defined in two different ways. The convex-roof construction~\cite{LeoneBittel2024} promotes the SRE to a faithful mixed-state magic monotone, whereas the extension introduced in Ref.~\cite{LeoneOlivieroHamma2022} has the advantage of naturally characterizing the free set $\mathrm{STAB0}$, which plays a central role in this work.
\begin{definition}[Stabilizer R\'enyi entropy]\label{def:sre}
Let $\rho$ be a (possibly mixed) quantum state. The stabilizer R\'enyi entropy of order $\alpha$ is defined as
\begin{align}
M_{\alpha}(\rho)\coloneqq S_{\alpha}(\Xi_{\rho})+S_{2}(\lambda_\rho)-n,
\end{align}
where $S_{\alpha}$ denotes the R\'enyi entropy of order $\alpha$, $\lambda_{\rho}$ denotes the spectrum of $\rho$, and $\Xi_{\rho}$ is the Pauli spectrum of $\rho$, that is, the probability distribution with components $\Xi_{\rho}(P)\coloneqq\frac{\tr^2(P\rho)}{d\tr(\rho^2)}$, where $P$ is a Pauli string. Remarkably, the stabilizer R\'enyi entropy faithfully characterizes the set $\mathrm{STAB0}$, namely
\begin{align}
M_{\alpha}(\rho)=0\iff \rho\in \mathrm{STAB0}.
\end{align}
\end{definition}

\section{A closed-form expression for non-local magic}
Let us begin by defining non-local magic formally.
\begin{definition}[Non-local magic] Let $A|B$ a bipartition of $n$ qubits. Let $\mathcal{M}$ any magic monotone and $\ket{\psi}$ a bipartite pure quantum state. The non-local magic of $\ket{\psi}$ is defined as
\begin{align}
    \mathcal{M}^{\mathrm{NL}}(\ket{\psi})\coloneqq \min_{U_A,U_B}\mathcal{M}(U_A\otimes U_B\ket{\psi})
\end{align}
where $U_A$ (resp. $U_B$) is a unitary operator on $A$ (resp. $B$).
\end{definition}
While, in general, the optimization over bipartite unitaries may be impossible to evaluate analytically and doubly exponentially difficult (in the number of qubits of the respective parts) to perform numerically, in the following theorem we identify a particular measure of magic, the log-stabilizer fidelity in \cref{def:stabfid}, for which the optimization admits a closed-form analytical solution. This constitutes one of the main results of this work.
\begin{theorem}[Analytical expression for the non-local log-stabilizer fidelity]\label{th:mainth}
       Let $\ket{\psi}=\sum_{i}\theta_i\ket{\phi_{i}}_A\ket{\chi_{i}}_B$ be a bipartite pure state in its Schmidt decomposition. Let $D_{\mathrm{stab}}$ be the log-stabilizer fidelity. Then, it holds that
    \begin{align}
        D^{\mathrm{NL}}_{\mathrm{stab}}(\ket{\psi})=-\log\max_{k \leq \min\{n_A,n_B\}}2^{-k}\left(\sum_{i=1}^{2^k}\theta_{i}^{\downarrow}\right)^2
    \end{align}
    where $\theta_i^\downarrow$ are the Schmidt coefficients of $\ket{\psi}$ rearranged in descending order.
    Moreover, the maximum is achieved by the state having the Schmidt basis equal to the computational basis. Yet, this maximum is not unique as it can be reached by other stabilizer bases. The proof is given in \cref{app:nlm}.
    \begin{proof}
        The proof strategy is to show that the r.h.s. upper and lower bounds $D_{\mathrm{stab}}^{\mathrm{NL}}$. Let us start from the lower bound. Let us denote $F_{\mathrm{stab}}^{\mathrm{NL}}(\ket{\psi})\coloneqq\max_{U_A,U_B}F_{\mathrm{stab}}(U_A\otimes U_B\ket{\psi})$. Let us choose $U_A^{*}$ and $U_B^{*}$ such that
        \begin{align}
            U_A^{*}\otimes U_B^{*}\ket{\psi}=\sum_i\theta_i^{\downarrow}\ket{i}_A\ket{i}_B
        \end{align}
        where $\{\ket{i}\}$ is the computational basis and the Schmidt spectrum is in descending order. Moreover, let us choose as stabilizer state $\ket{\sigma}=2^{-k/2}\sum_{i=1}^{2^k}\ket{i}\ket{i}\otimes \ket{0}_{n-k}$, i.e. a Bell pair on the first qubits and $\ket{0}$ on the rest. For each $k$, it holds that
        \begin{align}
            F_{\mathrm{stab}}^{NL}(\ket{\psi})\ge |\langle\sigma| U_A^{*}\otimes U_B^{*}\ket{\psi}|^2=\frac{1}{2^k}\left(\sum_{i=1}^{2^k}\theta_{i}^{\downarrow}\right)^2
        \end{align}
        Hence, we can also take the $\max$ over $k$. This shows the lower bound. For the upper bound, we make use of \cref{lem1} proved in \cref{app:nlm}. Since local unitaries do not change the Schmidt spectrum we have:
        \begin{align}
            \max_{U_A,U_B}\max_{\sigma}|\langle\sigma|U_A\otimes U_B|\psi\rangle|^2\le \max_{\sigma}\Big(\sum_{i}\lambda_i^\downarrow(\sigma)\theta_i^\downarrow\Big)^2\,.
        \end{align}
        where $\lambda_i(\sigma)$ are the Schmidt coefficients of $\sigma$.
        Every stabilizer state $\sigma$ has Schmidt spectrum $\lambda_i(\sigma)=2^{-k/2}$ for $i\in S_\sigma$ and $|S_\sigma|=2^k$ for some subset $S_\sigma$ of the indices of the Schmidt coefficients, depending on the stabilizer state. Hence:
        \begin{align}
            \max_{U_A,U_B}\max_{\sigma}|\langle\sigma|U_A\otimes U_B|\psi\rangle|^2\le\max_{k\le \min\{n_A,n_B\}}2^{-k}\Big(\sum_{i\le 2^k}\theta_i^\downarrow\Big)^2
        \end{align}
        proving the statement.
    \end{proof}
\end{theorem}

\begin{corollary}\label{cor:maxent}
    The non-local magic can be interpreted as the maximum overlap with a maximally entangled state of dimension $2^k$.
\end{corollary}

\begin{remark}
    The computation of the non-local magic requires $O(2^n)$ operations because it depends only on the Schmidt spectrum in a clear way and one needs to compute $n$ times the sums $\sum_{i=1}^{2^k}\theta_i$ for $k\in [n]$. Remarkably, the computation of non-local log-stabilizer fidelity is much cheaper then the computation of the log-stabilizer fidelity itself.
\end{remark}

\begin{remark} The non-local log-stabilizer fidelity becomes also accessible experimentally via the experimental reconstruction of the entanglement spectrum.

\end{remark}

\section{Non-local magic as magic of purification}
In this section, we make explicit the connection between non-local magic and the magic of purification. In analogy with the entanglement of purification~\cite{Terhal_2002}, the magic of purification is defined as the minimum magic attainable over all purifications of a mixed quantum state and, as we will show below, provides an alternative extension of the stabilizer resource theory from pure to mixed states.

\subsection{Magic of purification}
\begin{definition}[Magic of purification]\label{def:mop}
Let $\rho=\sum_i p\ketbra{\phi_i}{\phi_i}$ be a (possibly) mixed state. Let $\mathcal{M}$ be a magic monotone. Then the magic of purification is defined as:
\begin{align}
\mathcal{M}_{p}(\rho) := \min_{\ket{\psi}\,:\, \tr_E(\psi)=\rho}\mathcal{M}(\ket{\psi})
\end{align}
where the minimum is taken over the purification of $\rho$. $E$ is the ancillary system used to purify $\rho$.
\end{definition}
Magic of purification naturally induces a resource theory on $\operatorname{STAB0}$ as it is shown in the following.
\begin{proposition}[Resource theory of $\operatorname{STAB0}$]\label{prop:stab0} Magic of purification is both faithful and a good monotone on the set $\mathrm{STAB0}$:
\begin{itemize}
    \item $\mathcal{M}_p(\rho)=0$ if and only $\rho\in\operatorname{STAB0}$;
    \item $\mathcal{M}_{p}(\mathcal{E}(\rho))\le \mathcal{M}_p(\rho)$ for every linear free operation  $\mathcal{E}\in\mathcal{S}_0$ leaving invariant $\mathrm{STAB0}$.
\end{itemize}
\begin{proof}
We first prove faithfulness. By definition, $\rho\in\operatorname{STAB0}$ if and only if $\rho$ admits a purification that is a pure stabilizer state, and since $\mathcal{M}$ is a faithful pure-state monotone, $\mathcal{M}(\ket{\psi})=0$ if and only if $\ket{\psi}$ is a stabilizer state.
We now prove monotonicity. Clifford unitaries and the composition with
stabilizer extra qubits leave the magic of purification invariant, while a
non-selective measurement in the computational basis can be implemented
by coupling the measured qubit to a stabilizer auxiliary system through a CNOT
gate and subsequently tracing out the auxiliary system. Hence, it is sufficient
to prove monotonicity under partial trace. Let $\rho$ be a (possibly) mixed state and $\ket{\psi_{\rho}}=\sum_{i}\ket{\phi_i}\otimes\ket{i}_E$ (with $\langle\phi_i|\phi_i\rangle$ the eigenvalues of $\rho$) be the optimal purification for \cref{def:mop}. Performing the partial trace we have $\rho_X=\tr_{\bar{X}}\rho$. Due to linearity of the partial trace, we have that $\rho_X=\sum_{i}\tr_{\bar{X}}\ketbra{\phi_i}{\phi_i}$. Hence, clearly, a possible purification of $\rho_X$ is $\ket{\psi_{\rho_X}}=\sum_{i}\ket{\phi_{i}}_{X\bar{X}}\otimes \ket{i}_{E}$ with auxiliary systems $E\cup X$. Hence, we have that $\mathcal{M}_{p}(\rho_X)\le \mathcal{M}(\ket{\psi_{\rho_X}})=\mathcal{M}_p(\rho)$.
\end{proof}

\end{proposition}

\begin{remark}
   The stabilizer entropy introduced in \cref{def:sre} plays a peculiar role in the resource theory of magic of purification. Indeed, although counterexamples exist to its monotonicity under partial trace, and hence it cannot by itself be regarded as a magic-of-purification monotone without the optimization introduced in \cref{def:mop}, the SRE nevertheless \textit{faithfully} identifies the free set $\mathrm{STAB0}$ underlying the resource theory of magic of purification.

\end{remark}

Along the lines of entanglement theory, we now define the magic of formation as the standard \textit{convex-roof extension} of a pure-state magic monotone to mixed states within the corresponding convex hull, which, we recall, is denoted by $\mathrm{STAB1}$ throughout this work.

\begin{definition}[Magic of formation]\label{def:mof} Let $\rho$ be a (possibly) mixed state. Then, the magic of formation is defined as
\begin{align}
\mathcal{M}_f(\rho):=\inf_{\{p_i,\phi_i\}}\left\{\sum_{i}p_i\mathcal{M}(\phi_i)\,:\, \rho=\sum_ip_i\phi_i\right\}
\end{align}
\end{definition}
It is well known \cite{ChitambarGour2019} that the magic of formation induces a resource theory on $\operatorname{STAB1}$ \cite{LeoneBittel2024}, provided that the measure $\mathcal{M}$ used in the optimization obeys strong monotonicity (\cref{def:stabilizermonotones}).
The natural question that arises at this point is: is there a relation between $\mathcal{M}_p$ and $\mathcal{M}_f$, given that for pure states they coincide? The answer is Yes, as stated by the following theorem:
\begin{theorem}[Magic of formation lower bounds the magic of purification]\label{thm:mfmp} Given a pure state magic monotone $\mathcal{M}$ obeying strong monotonicity, for any state $\rho$ then the following bound holds
\begin{align}
\mathcal{M}_f(\rho)\le \mathcal{M}_{p}(\rho)
\end{align}
\begin{proof}
Given $\rho = \sum_i p_i \ketbra{\phi_i}{\phi_i}$, let $\ket{\psi_{\rho}}$ be an optimal purification in the sense of \cref{def:mop}, i.e. $\mathcal{M}_p(\rho)=\mathcal{M}(\ket{\psi_{\rho}})$. Then $\ket{\psi_{\rho}}=\sum_{i}\sqrt{p_i}\ket{\phi_i}\otimes\ket{\psi_i}_E$, where $\ket{\psi_i}_E$ are not stabilizer states in general. We can decompose them in the computational basis as $\ket{\psi_i}_E=\sum_{j}\alpha_{ij}\ket{j}_E$ and rewrite the purification as
\begin{align}
\ket{\psi_{\rho}}=\sum_{ij}\alpha_{ij}\sqrt{\lambda_i}\ket{\phi_{i}}\otimes \ket{j}_E=\sum_{j}\ket{\tilde{\phi}_j}\otimes  \ket{j}_E
\end{align}
Notice that $\sum_{j}\ketbra{\tilde{\phi}_j}{\tilde{\phi}_j}$ is a valid convex decomposition of $\rho$. Now, we apply the following stabilizer operation: we measure the subsystem $E$ in the computational basis, we obtain the state $\ket{\tilde{\phi}_j}$ with probability $q_{j}=\langle\tilde{\phi}_j|\tilde{\phi}_j\rangle$. Thanks to the strong monotonicity of $\mathcal{M}$, we know that it holds that
\begin{align}
\mathcal{M}(\ket{\psi_{\rho}})\ge \sum_{j}q_j\mathcal{M}\left(\frac{\ket{\tilde{\phi}_j}}{\sqrt{q_j}}\right)
\end{align}
However, given that $\sum_{j}\ketbra{\tilde{\phi}_j}{\tilde{\phi}_j}=\sum_{j}q_j\frac{\ketbra{\tilde{\phi}_j}{\tilde{\phi}_j}}{q_j}$, we also have
\begin{align}
\sum_{j}q_j\mathcal{M}\left(\frac{\ket{\tilde{\phi}_j}}{\sqrt{q_j}}\right)\ge \mathcal{M}_{f}(\rho)
\end{align}
since $\mathcal{M}_f(\rho)$ optimizes over all the possible convex decomposition. The latter proves the theorem.
\end{proof}

\end{theorem}

In the following theorem, we characterize the log-stabilizer fidelity of purification. Once again, this quantity stands out as a particularly convenient magic monotone, as it admits explicit expressions that make its evaluation straightforward.

\begin{theorem}[Log-stabilizer fidelity of purification]\label{lem:ancilla}
Given the log-stabilizer fidelity, the optimization procedure required in \cref{def:mop} can be conducted with an auxiliary system $E$ of size $|E|= n$. Moreover, the log-stabilizer fidelity of purification can be explicitly written as
\begin{align}
    D_{\mathrm{stab}}^{\mathrm{pur}}(\rho)\coloneqq-\log \max_{\sigma\in\mathrm{STAB0}}F(\rho,\sigma)
\end{align}
\begin{proof}
    This just follows from the Uhlmann theorem. More explicitly,  Uhlmann theorem says the following: the fidelity between two (possibly mixed) states $\rho,\sigma$ is the maximum overlap between a purification of $\sigma$ over all possible purification of $\rho$ performed with an ancillary system $|E|\ge\log\max(\rank(\rho),\rank(\sigma))$. In other words, for any ancillary system $E$ (with appropriate size) there exist a maximum of the overlap between the two states. These maxima are linked by an isometry between subsystems of different dimensions.

    In the definition of stabilizer fidelity, for each subsystem of appropriate size and for any stabilizer state, we can write
    \begin{align}
    \max_{\psi\,:\, \tr_E(\psi)=\rho,|E|=m}|\langle\psi|\sigma\rangle|^2=F(\rho,\tr_E\sigma)
    \end{align}
    by the Uhlmann theorem. Since pure stabilizer states are a finite set, their partial traces on $m$ qubits, i.e. $\tr_E\sigma$, are a finite set. Thus, varying $m\in[\lceil\log\rank(\rho)\rceil,\infty)$, we can therefore equivalently write the log-stabilizer fidelity of purification $D_{\mathrm{stab}}^{\mathrm{pur}}$ as
    \begin{align}
    D_{\mathrm{stab}}^{\mathrm{pur}}(\rho)=-\log\max_{\sigma\in \operatorname{STAB}_0}F(\rho,\sigma)\label{eq:alternativestabfidelity}
    \end{align}
    Since the maximal rank of a matrix on $n$ qubits is $2^n$, we trivially have that we can optimize over the set of size $|E|=n$. This concludes the proof.
\end{proof}
\end{theorem}

\begin{remark}
   The optimization form of the log-stabilizer fidelity of purification in \cref{lem:ancilla} makes it evident that the magic of purification induces a resource theory over the subset $\mathrm{STAB0}\subset \mathrm{STAB1}$, as the corresponding optimization takes the standard form of a distance-based resource measure, namely an optimization over the set of free states~\cite{ChitambarGour2019}.

\end{remark}

\begin{corollary}\label{cor:fpn}
    It holds that $D_{\mathrm{stab}}^{\mathrm{pur}}(\rho)\le n$ for any $n$ qubit quantum state $\rho$.
    \begin{proof}
        Thanks to Eq.~\eqref{eq:alternativestabfidelity}, we can upper bound it with the maximally mixed state $D_{\mathrm{stab}}^{\mathrm{pur}}(\rho)\le-\log F(\rho, I/2^n)\le n$, which concludes the proof.
    \end{proof}
\end{corollary}

\subsection{Non-local magic as the minimal magic of purification across unitary orbits}

The non-local magic of a bipartite pure state is a spectral quantity, depending only on its Schmidt coefficients; see \cref{th:mainth}. It is therefore natural to ask how it relates to the magic of the corresponding reduced state, and, in particular, which states minimize the magic along a given unitary orbit. In the following, we establish the second main result of this work: the minimum magic of purification of a state $\rho$ over its unitary orbit $U\rho U^{\dagger}$ coincides with the non-local magic of \textit{any} purification of $\rho$. We also characterize the states that attain this minimum. The next theorem establishes a relationship between magic of purification and non-local magic.

\begin{theorem}[]\label{thm:orbitmin}
The minimal log-stabilizer fidelity of purification in the unitary orbit of a density operator equals
the non-local magic of its purification: for every quantum state
$\rho_A$ and every purification
$\ket{\rho}_{AB}$,
\begin{align}
\min_{U_A\in\mathcal{U}_A} D_{\mathrm{stab}}^{\mathrm{pur}}\big(U_A\rho_A U_A^\dagger\big)
=D_{\mathrm{stab}}^{\mathrm{NL}}\big(\ket{\rho}_{AB}\big)\,.
\end{align}
\begin{proof} Let us denote the eigenvalues of the reduced density matrix $\rho_A$ on $A$ as
$p_1^\downarrow\ge p_2^\downarrow\ge\dots$, we set $\lambda_i:=\sqrt{p_i}$,
$r:=\rank\rho_A$, and
\begin{align}
f(k):=2^{-k}\Big(\sum_{i=1}^{2^k}\lambda_i^\downarrow\Big)^{2}.
\end{align}
Then, notice that every $\sigma\in STAB_0$ is of the form $\sigma=\Pi_S/2^k$ with $\Pi_S$ a
rank-$2^k$ \emph{stabilizer code} projector, equivalently
$\sigma=C\,(\Pi_0^k/2^k)\,C^\dagger$ with $C$ Clifford and
$\Pi_0^k=\sum_{i<2^k}\ketbra{i}{i}$. (Note that not every $2^k$-element subset of
a stabilizer basis spans a stabilizer code; the subset must be an affine
subspace). Using \cref{lem:ancilla} and writing $\sigma=C(\Pi_0^k/2^k)C^\dagger$, the
unitary invariance of the fidelity lets us absorb the Clifford $C$ into $U_A$, so
that with $V:=C^\dagger U_A$ ranging over all of $\mathcal{U}_A$,
\begin{align}
\max_{U_A}\max_{\sigma\in\mathrm{STAB0}}F(U_A\rho_AU_A^\dagger,\sigma)
=\max_{k}\ \max_{V\in\mathcal{U}_A} F\Big(V\rho_AV^\dagger,\tfrac{\Pi_0^k}{2^k}\Big)
=\max_{k}\ \max_{V} 2^{-k}\big\|V\sqrt{\rho_A}V^\dagger\,\Pi_0^k\big\|_1^2,
\end{align}
where we used $\sqrt{\Pi/2^k}=\Pi/2^{k/2}$ and
$F(\tau,\omega)=\norm{\sqrt{\tau}\sqrt{\omega}}_1^2$. The operator
$V\sqrt{\rho_A}V^\dagger$ is positive with eigenvalues $\lambda_i$, and
$\rank\Pi_0^k=2^k$, so \cref{lem:compression} gives
$\big\|V\sqrt{\rho_A}V^\dagger\Pi_0^k\big\|_1\le\sum_{i\le2^k}\lambda_i^\downarrow$,
with equality attained by any $V$ mapping the $2^k$ leading eigenvectors of
$\rho_A$ onto $\operatorname{supp}\Pi_0^k$. Hence
$\max_{U_A}\max_{\sigma\in\mathrm{STAB0}}F(U_A\rho_AU_A^\dagger,\sigma)=\max_k f(k)$. Thanks to \cref{th:mainth}, the two maxima therefore coincide for \emph{every} purification. This proves the theorem.
\end{proof}
\end{theorem}

%[Characterisation of the minimisers]
\begin{remark}\label{cor:minimisers}
A state $\rho_A$ attains $\min_{U_A\in\mathcal{U}_A} D_{\mathrm{stab}}^{\mathrm{pur}}\big(U_A\rho_A U_A^\dagger\big)$ if and only if there exists $k$
with $f(k)=\max_{k'}f(k')$ such that the span of the eigenvectors of $\rho_A$
belonging to its $2^{k}$ largest eigenvalues is a stabilizer code subspace.
\end{remark}

\begin{remark}\label{cor:stab1exists}
Every unitary orbit contains a magic minimum belonging to $\mathrm{STAB1}$. Indeed
$\rho'=\sum_i p_i^\downarrow\ketbra{i}{i}$, diagonal in a stabilizer basis, is a
log-stabilizer fidelity minimiser, since $\operatorname{span}\{\ket{i}\}_{i<2^k}$ is a stabilizer code for
every $k$; and $\rho'$ is a mixture of pure stabilizer states.
\end{remark}

\begin{remark}
\label{remark:minimizerNL}
    As a consequence of Remark~\ref{cor:stab1exists}, the state in Eq.~(53) of Ref.~\cite{Cao_2025}, used there to estimate the non-local magic, attains the minimum of the log-stabilizer fidelity over local unitaries.
\end{remark}

\section{Non-local magic, Flatness and semicausality}\label{sec:causality}
 We now show that non-local magic is the resource needed for \emph{graded}
signalling. Vanishing non-local magic is equivalent to a flatness condition on the
entanglement spectrum \cite{Tirrito2024,Cao_2025}, and
thus is the signature of stabilizer entanglement spectra; equivalently, the reduced state can be rotated into
$\mathrm{STAB0}$.
Through the Choi--Jamio{\l}kowski isomorphism we show that anti-flatness controls the signalling of the  associated quantum operation: stabilizer resources
permit only all-or-nothing causal influence, and any continuously graded
influence costs magic.

Let us associate to a pure state $\ket{\psi}_{AB}$  both a POVM and a quantum map.
Consider a system of $n$ qubits with $d=2^n$, and let
$\ket{\Omega}=d^{-1/2}\sum_{i=1}^{d}\ket{i}_A\ket{i}_B$ be the maximally
entangled state. Every bipartite pure state can be written as
$\ket{\psi}_{AB}=(K\otimes\id)\ket{\Omega}$ for a unique operator $K$ obeying
$\tr(K^\dagger K)=d$, and we let $\mathcal{E}_K(\rho)=K\rho K^\dagger$ be the
associated completely positive map. Denoting by $s_1\ge s_2\ge\dots$ the singular
values of $K$, the Schmidt coefficients of $\ket{\psi}$ are $\theta_i=s_i/\sqrt d$.

The map $\mathcal{E}_K$ is not trace preserving, so to speak of it operationally
we must realise it as one branch of an instrument. Set
\ba
\hat K:=K/s_1,\qquad\text{so that}\qquad \hat K^\dagger\hat K\le\id ,
\ea
with largest eigenvalue $1$. Then $\{\hat K,\ \sqrt{\id-\hat K^\dagger\hat K}\}$
is a two-outcome instrument, and we call
\ba
p_\psi(\rho):=\tr\big(\hat K^\dagger\hat K\rho\big)\in[0,1]
\label{eq:successprob}
\ea
the {success probability} of the filter, that is, the probability that the branch
$\hat K$ occurs on input $\rho$.
From the singular value decomposition $\hat K=\sum_i(s_i/s_1)\ketbra{u_i}{v_i}$, we construct
the effect
\begin{align}
E:=\hat K^\dagger\hat K=\sum_i\frac{s_i^{2}}{s_1^{2}}\ketbra{v_i}{v_i}
=\sum_i\frac{\theta_i^{2}}{\theta_1^{2}}\ketbra{v_i}{v_i}
\end{align}
The two-outcome POVM $\{E,\id-E\}$ is described by the following

\begin{lemma}[]\label{lem:choi}
With the above identification, $\rho_A=\tfrac1d KK^\dagger$ and
$\rho_B=\tfrac1d\overline{K^\dagger K}$; in particular the Schmidt spectrum
$\{\theta_i^2\}$ of $\ket{\psi}$ is the spectrum of $K^\dagger K$ rescaled by
$1/d$. Consequently there are three cases:
\begin{itemize}
\item[(i)] $\rho_A=\id/d$ (flat of full rank, i.e.\ maximally entangled)
$\iff \hat K$ is unitary $\iff\mathcal{E}_{\hat K}$ is trace preserving. Then
$p_\psi\equiv1$: the branch occurs with certainty and its occurrence carries no
information about the input.
\item[(ii)] $\rho_A=\Pi_{\rm out}/2^{k}$ for some rank-$2^k$ projector
$\Pi_{\rm out}$ $\iff$ $\hat K$ is a partial isometry, i.e.\ $E=\Pi_{\rm in}$ is
a rank-$2^k$ projector. Here
$\Pi_{\rm in}=\sum_{i\le2^{k}}\ketbra{v_i}{v_i}$ and
$\Pi_{\rm out}=\sum_{i\le2^{k}}\ketbra{u_i}{u_i}$ act on the input and on the
output space respectively and are exchanged by $\hat K$; they need not coincide.
Then $p_\psi(\rho)=\tr(\Pi_{\rm in}\rho)$, so $\{\Pi_{\rm in},\id-\Pi_{\rm in}\}$
is a \emph{sharp}, that is projective, two-outcome measurement: the filter
accepts with certainty every state supported in $\operatorname{ran}\Pi_{\rm in}$ with equal probabilities.

\item[(iii)]  $E$ has an eigenvalue strictly between $0$ and $1$: the
effect is \emph{unsharp}, and $p_\psi$ separates states \emph{inside} the
support. We call such a filter \emph{graded}. For instance
$\hat K=\operatorname{diag}(1,\sqrt{\eta})$ with $0<\eta<1$ accepts $\ket{0}$
with certainty and $\ket{1}$ with probability $\eta$, although both lie in its
support.
\end{itemize}
\begin{proof}
Writing $\ket{\psi}=d^{-1/2}\sum_i K\ket{i}\otimes\ket{i}$ and tracing out either
factor gives the two marginals; normalisation of $\ket{\psi}$ is
$\tr(K^\dagger K)=d$. Since $KK^\dagger$ and $K^\dagger K$ are isospectral, both
marginals have spectrum $\{s_i^2/d\}=\{\theta_i^2\}$, which is the Schmidt
spectrum. As $E$ has eigenvalues $\theta_i^2/\theta_1^2$, flatness of the nonzero
part of $\{\theta_i^2\}$ is equivalent to all these eigenvalues lying in
$\{0,1\}$, that is to $E$ being a projector. The three cases are therefore
precisely $E=\id$, $E$ a projector different from $\id$, and
$\operatorname{spec}(E)\cap(0,1)\neq\emptyset$; Eq. \ref{eq:successprob} then gives
 $p_\psi$ in each case.
\end{proof}
\end{lemma}

As the Choi isomorphism is between \emph{states} and \emph{operators} in order to speak of free
operations we must say which operators correspond to stabilizer Choi states. The following lemma establishes that

\begin{lemma}[]\label{lem:stabchoi}
$\ket{\psi}_{AB}$ is a pure stabilizer state if and only if
$K\propto C\,\Pi_S$, with $C$ a Clifford unitary and $\Pi_S$ the projector onto a
stabilizer code. Equivalently, $\mathcal{E}_{\hat K}$ is a Clifford unitary
followed by, or preceded by, a stabilizer projection.
\begin{proof}
($\Leftarrow$) $\ket{\Omega}$ is a stabilizer state, $\Pi_S\otimes\id$ is a
product of projectors onto Pauli eigenspaces, and stabilizer states are mapped to
(unnormalised) stabilizer states by such projections and by Cliffords; hence
$(C\Pi_S\otimes\id)\ket{\Omega}$ is a stabilizer state.
($\Rightarrow$) By the bipartite normal form, a pure stabilizer state on $AB$ is
local-Clifford equivalent to $k$ Bell pairs tensored with a product stabilizer
state, i.e.\ $\ket{\psi}=(C_A\otimes C_B)\big(\ket{\Phi_{2^k}}\otimes
\ket{0}^{\otimes(n-k)}_A\ket{0}^{\otimes(n-k)}_B\big)$. Using
$(\id\otimes M)\ket{\Omega}=(M^{T}\otimes\id)\ket{\Omega}$ the Clifford $C_B$ may
be transferred to the $A$ side, and the seed state is $(\Pi\otimes\id)\ket\Omega$
for $\Pi$ the stabilizer code projector onto the first $k$ qubits; collecting the
Cliffords and using that $C\Pi C^\dagger$ is again a stabilizer code projector
gives $K\propto C\Pi_S$.
\end{proof}
\end{lemma}

By \cref{lem:stabchoi} the marginals of stabilizer Choi states are of the form
$\Pi_S/2^{k}$, so stabilizer resources populate only the cases (i) and (ii) of
\cref{lem:choi}, and never (iii).

\begin{corollary}\label{cor:gradedfilters}
Let $\ket{\psi}_{AB}$ be the Choi state of a filter with effect $E$, and let
$\rho_A$ be its reduced state. The operations whose Choi state is a stabilizer
state are, up to normalisation, the Clifford unitaries ($\Pi_S=\id$, case (i))
and the sharp stabilizer projections ($\Pi_S\neq\id$, case (ii)); in both the
effect $E$ is a projector. More generally, the following are equivalent:
\begin{itemize}
\item[(a)] the effect $E$ is a projector, i.e.\ the filter is trivial or sharp
(cases (i)--(ii) of \cref{lem:choi});
\item[(b)] the unitary orbit of $\rho_A$ meets $\mathrm{STAB0}$;
\item[(c)] $\min_{U_A\in\mathcal{U}_A}
D^{\mathrm{pur}}_{\mathrm{stab}}\big(U_A\rho_AU_A^\dagger\big)=0$;
\item[(d)] $D^{\mathrm{NL}}_{\mathrm{stab}}(\ket{\psi})=0$.
\end{itemize}
Consequently every \emph{unsharp} filter costs magic: no unitary rotation of its reduced state lies in
$\mathrm{STAB0}$, and
\begin{align}
D_{\mathrm{stab}}(\ket{\psi})\ \ge\ D^{\mathrm{NL}}_{\mathrm{stab}}(\ket{\psi})
=\min_{U_A\in\mathcal{U}_A}D^{\mathrm{pur}}_{\mathrm{stab}}\big(U_A\rho_AU_A^\dagger\big)
\ >\ 0 .
\end{align}
\begin{proof}
Notice that the statement (a) is \cref{lem:stabchoi} together with \cref{lem:choi}. For the equivalences,
(a)$\Leftrightarrow$(b): by \cref{lem:choi}, $E$ is a projector iff the Schmidt
spectrum of $\ket{\psi}$ is $2^k$-flat, which is equivalent
to the unitary orbit of $\rho_A$ meeting $\mathrm{STAB0}$.
(b)$\Leftrightarrow$(c): by \cref{lem:ancilla},
$D^{\mathrm{pur}}_{\mathrm{stab}}(\rho)=-\log\max_{\sigma\in\mathrm{STAB0}}F(\rho,\sigma)$
vanishes iff $F(\rho,\sigma)=1$ for some $\sigma\in\mathrm{STAB0}$, i.e.\ iff
$\rho\in\mathrm{STAB0}$; moreover it is a continuous function of $\rho$, being
the maximum of finitely many continuous functions, so its minimum over the
compact unitary orbit is attained. Hence the minimum vanishes iff some point of
the orbit lies in $\mathrm{STAB0}$.
(c)$\Leftrightarrow$(d) is \cref{thm:orbitmin}. Finally, for an unsharp filter
both (a) and (d) fail and
$D_{\mathrm{stab}}(\ket{\psi})\ge D^{\mathrm{NL}}_{\mathrm{stab}}(\ket{\psi})$
because the non-local magic is a minimum over the local-unitary orbit, which
contains $\ket{\psi}$ itself.
\end{proof}
\end{corollary}

We have established that magic is the resource needed for graded signalling. The case (ii) is within the stabilizer resources, and, though signalling, it is   all or nothing. Case (iii) requires non-local non-stabilizerness and  is graded
causal influence.
%======================================================================
\section{Clifford dilations and local injection of magic}\label{sec:dilation}
%======================================================================
One way of looking at the combined effect of entanglement and magic is to look at their roles in producing a non trivial entanglement spectrum. Clifford entangling operations populate the Schmidt spectrum, but in a flat (and thus trivial) way. Then, one can inject magic locally. At first, this injection does not change the entanglement spectrum. However, by further Clifford non-local interaction, magic is scrambled around and flexes the entanglement spectru making it non trivial. In this way non-local magic produces non-flat, non trivial entanglement spectra. This section is devoted in finding exact quantitative results for this mechanism. In the spirit of understanding non-local magic under the lens of purification, let us start with describing quantum processes as quantum channels in their Stinespring representation.

By Stinespring's theorem, every quantum channel $\mathcal N$ on a system $S$ of
$n$ qubits ($d=2^n$) can be written as
$\mathcal N(\rho)=\tr_B[W(\rho\otimes\dyad{0}_B)W^\dagger]$ for some unitary $W$
on $S\otimes B$. In general $W$ is a non-Clifford entangling unitary. The question
we address in this section is how much of the non-Clifford resources of $W$ can
be moved into \emph{local} operations, leaving a Clifford interaction
between system and ancilla, and how much magic such a local injection costs. The
guiding example is a projective measurement: the interaction with the apparatus
is a (generalized) CNOT, and the only non-Clifford ingredient is the choice of
measurement basis, which acts on the system alone. While in \cref{sec:causality}
applied the Choi--Jamio{\l}kowski dictionary to a single Kraus operator (a filter), now instead 
 we consider complete quantum channels. In \cref{subsec:steering} we also use the Clifford dilation of projective measurements to analyse quantum steering.

\subsection{Definitions}

For a Clifford unitary $C$ on $S\otimes W$, with $W$ a register of $m$ qubits,
and a pure state $\tau$ of $B$, we define the channel
\begin{align}
\Phi_{C,\tau}(\rho)\coloneqq\tr_B\!\left[C(\rho\otimes\tau)C^\dagger\right].
\end{align}
Unitary channels are denoted by the adjoint action $\mathrm{Ad}_U(\rho)\coloneqq U\rho U^\dagger$.

\begin{definition}[Clifford dilation with locally injected magic]\label{def:clifforddilation}
A channel $\mathcal N$ on $S$ admits a \emph{Clifford dilation with locally injected
magic} if there exist   a Clifford unitary channel  $C$ on a suitable $S\otimes B$, a pure
state $\tau$ on $B$, and unitaries $U,V$ on $S$ such that
\begin{align}
\mathcal N=\mathrm{Ad}_V\circ\Phi_{C,\tau}\circ\mathrm{Ad}_U .\label{eq:clifforddilation}
\end{align}
We denote by $\mathfrak D$ the set of such channels and by
$\mathfrak D_0\subseteq\mathfrak D$ the subset for which $\tau$ can be chosen to be a
stabilizer state and, without loss of generality, $\tau=\dyad{0}^{\otimes m}$. The \emph{dilation magic cost} of
$\mathcal N$ is
\begin{align}
D_{\mathrm{dil}}(\mathcal N)\coloneqq\inf\left\{D_{\mathrm{stab}}(\tau)\;:\;
\mathcal N=\mathrm{Ad}_V\circ\Phi_{C,\tau}\circ\mathrm{Ad}_U\right\}
\end{align}
\end{definition}
Notice that Eq.\ref{eq:clifforddilation} is the most general
dilation in which arbitrary non-Clifford operations act locally on $S$ and on $B$ while the interaction between them is Clifford. In $\mathfrak D_0$
all the magic is injected on the system.
Also notice that Pauli measurements with Pauli
feed-forward are included, since controlled Paulis are Clifford. Adaptive
non-Pauli corrections, as used in magic-state injection, are not; with those,
magic states together with Clifford operations become (approximately) universal.

\subsection{Projective measurements}

Let $\{\ket{\sigma_k}\}_{k=0}^{d-1}$ be an orthonormal basis of $S$,
$\Pi_k^\sigma\coloneqq\dyad{\sigma_k}$, and let
$\mathcal D_\sigma(\rho)\coloneqq\sum_k\Pi^\sigma_k\rho\,\Pi_k^\sigma$ be the
corresponding complete dephasing in that basis. Take $B$ to be a copy of $S$ and let the generalized CNOT gate be
\begin{align}
C\coloneqq\sum_{k}\dyad{k}\otimes X^{k}
\end{align}
where $X^{k}\coloneqq X^{k_1}\otimes\dots\otimes X^{k_n}$, so that
$C\ket{j}\ket{k}=\ket{j}\ket{k\oplus j}$ with $\oplus$ the  sum $\mod 2$.

\begin{proposition}[Clifford dilation of a projective measurement]\label{prop:rankonedilation}
Let $\tilde U$ be a unitary on $S$ with $\tilde U\ket{k}=\ket{\sigma_k}$ and set
$W^\sigma\coloneqq(\tilde U\otimes\id_E)\,C\,(\tilde U^\dagger\otimes\id_E)
=\sum_k\Pi_k^\sigma\otimes X^k$. Then
\begin{align}
\mathcal D_\sigma(\rho)=\tr_E\!\left[W^\sigma(\rho\otimes\dyad{0})W^{\sigma\dagger}\right],
\end{align}
that is, $\mathcal D_\sigma=\mathrm{Ad}_{\tilde U}\circ\Phi_{C,\ket0}\circ\mathrm{Ad}_{\tilde U^\dagger}\in\mathfrak D_0$.
\end{proposition}
\begin{proof}
Writing $\ket\psi=\sum_k\psi_k^\sigma\ket{\sigma_k}$, we have
$W^\sigma\ket\psi\ket0=\sum_k\psi^\sigma_k\ket{\sigma_k}\ket{k}$, and tracing out
$E$ gives $\sum_k|\psi_k^\sigma|^2\Pi^\sigma_k=\mathcal D_\sigma(\psi)$. The claim
extends to mixed $\rho$ by linearity.
\end{proof}

The unitary $W^\sigma$ is not Clifford unless $\{\ket{\sigma_k}\}$ is a stabilizer
basis but it is a Clifford interaction dressed by local non-Clifford unitaries. Note
also that $\tilde U$ is fixed only up to $\tilde U\ket k=e^{i\alpha_k}\ket{\sigma_{\pi(k)}}$.
To go beyond rank-one measurements, we need a necessary condition for membership in
$\mathfrak D_0$. Let $\ket{\Omega}_{SS'}\coloneqq d^{-1/2}\sum_i\ket{i}\ket{i}$ be the
maximally entangled state  on $S$ and  an ancilla $S'$, and
$J_{\mathcal N}\coloneqq(\mathcal N\otimes\id_{S'})(\dyad{\Omega})$ its Choi state.
The following lemma is the channel counterpart of \cref{lem:stabchoi}.

\begin{lemma}[Choi states of $\mathfrak D_0$]\label{lem:choiD0}
Let $\mathcal N=\mathrm{Ad}_V\circ\Phi_{C,\tau}\circ\mathrm{Ad}_U$. The pure state
\begin{align}
\ket{J}_{SS'B}\coloneqq(V_S\otimes U^{T}_{S'}\otimes\id_B)\,C_{SB}\,\ket{\Omega}_{SS'}\ket{\tau}_B
\label{eq:choipurification}
\end{align}
purifies $J_{\mathcal N}$. If $\tau$ is a stabilizer state, then
$J_{\mathcal N}=(V\otimes U^T)\,\sigma\,(V\otimes U^T)^\dagger$ for some
$\sigma\in\mathrm{STAB0}$ on $SS'$. In particular, the spectrum of $J_{\mathcal N}$
is flat: $\mathrm{spec}(J_{\mathcal N})=\{2^{-r}\}$ with multiplicity $2^{r}$ for some $r$.
\end{lemma}
\begin{proof}
Use $(U\otimes\id)\ket\Omega=(\id\otimes U^T)\ket\Omega$ and the fact that $U^T_{S'}$
commutes with $C_{SE}$ and $V_S$. If $\tau$ is a stabilizer state, then
$C\ket\Omega\ket\tau$ is a pure stabilizer state, so its marginal $\sigma$ on $SS'$
admits a stabilizer purification and hence lies in $\mathrm{STAB0}$, i.e.\ it is a normalized stabilizer-code projector $\Pi_{\mathcal G}/2^r$.
\end{proof}

\begin{proposition}[Degenerate projective measurements]\label{prop:degeneratedilation}
Let $\{P_a\}_{a=1}^{M}$ be a complete set of orthogonal projectors on $S$ and
$\mathcal D_P(\rho)\coloneqq\sum_aP_a\rho P_a$. Then $\mathcal D_P\in\mathfrak D_0$ if
and only if $M=2^{j}$ and $\mathrm{rank}\,P_a=2^{n-j}$ for all $a$. In that case an
ancilla of $j$ qubits suffices.
\end{proposition}
\begin{proof}
($\Leftarrow$) Label the outcomes by $x\in\{0,1\}^n$. Since the ranks are equal,
there is a unitary $U$ with $U(\dyad{x}\otimes\id_{n-j})U^\dagger=P_x$. With
$C$ the CNOT gate, the channel $\Phi_{C,\ket0}$ maps
 $\rho\mapsto \Phi_{C,\ket0}(\rho)=\sum_x(\dyad x\otimes\id)\rho(\dyad x\otimes\id))$, and
$\mathcal D_P=\mathrm{Ad}_U\circ\Phi_{C,\ket0}\circ\mathrm{Ad}_{U^\dagger}$.
($\Rightarrow$) We have
$J_{\mathcal D_P}=\sum_a(P_a\otimes\id)\dyad\Omega(P_a\otimes\id)$. The vectors
$(P_a\otimes\id)\ket\Omega$ are mutually orthogonal with squared norms
$\mathrm{rank}(P_a)/d$, so $\mathrm{spec}(J_{\mathcal D_P})=\{\mathrm{rank}(P_a)/d\}_a$.
By \cref{lem:choiD0} this spectrum is flat, $\mathrm{rank}(P_a)/d=2^{-j}$ for all
$a$, whence $M=2^j$.
\end{proof}
NOtice that \cref{prop:rankonedilation} corrresponds to the case $j=n$. By contrast, a two-outcome
measurement $\{\dyad{0^n},\id-\dyad{0^n}\}$ with $n\ge2$ admits no Clifford
dilation with magic injected on the system only.

\subsection{Clifford converts Local magic into non-local magic}

The dilation of \cref{prop:rankonedilation} applied to $\ket\psi$ produces the
system--apparatus state $\sum_k\psi_k^\sigma\ket{\sigma_k}\ket{k}$. Its Schmidt
coefficients are $\sqrt{p_k}$, where $p_k=|\psi_k^\sigma|^2$ are the Born
probabilities. A Clifford interaction therefore converts the local magic of the
input into non-local magic between system and apparatus. Let us now give a quantitative result. For a probability vector $\bm p$ we write
\begin{align}
D^{\mathrm{NL}}_{\mathrm{stab}}(\bm p)\coloneqq-\log\max_{k\ge0}\,2^{-k}\Big(\sum_{i=1}^{2^k}\sqrt{p^\downarrow_i}\Big)^2 ,
\end{align}
with $p_i^\downarrow=0$ beyond the support. The partial sums saturate once
$2^k\ge\mathrm{rank}(\bm p)$, so the maximum is attained at some
$k\le\lceil\log_2\mathrm{rank}(\bm p)\rceil\le\min\{n_A,n_B\}$. By \cref{th:mainth},
this expression is therefore the non-local magic of every bipartite pure state with
Schmidt coefficients $\sqrt{p_i}$, and by \cref{thm:orbitmin} it also equals
$\min_{U}D^{\mathrm{pur}}_{\mathrm{stab}}(U\rho U^\dagger)$ for any $\rho$ with
$\mathrm{spec}(\rho)=\bm p$. We now show a useful lemma for non-negative vectors:

\begin{lemma}[]\label{lem:nonnegfid}
Let $\ket\phi=\sum_{x\in\{0,1\}^m}c_x\ket x$ with $c_x\ge0$. Then
\begin{align}
F_{\mathrm{stab}}(\ket\phi)=\max_{\mathcal A}\;2^{-\dim\mathcal A}\Big(\sum_{x\in\mathcal A}c_x\Big)^2,
\end{align}
where the maximum is over   $\mathcal A\subseteq\mathbb \{0,1\}^m$.
\end{lemma}
\begin{proof}
Every pure stabilizer state has the form
$2^{-k/2}\sum_{x\in\mathcal A}i^{\ell(x)}(-1)^{q(x)}\ket x$, with $\mathcal A$  of dimension $k$, $\ell$ linear and $q$ quadratic.
By the triangle inequality its overlap with $\ket\phi$ is at most
$2^{-k/2}\sum_{x\in\mathcal A}c_x$. The bound is attained by the stabilizer state
$2^{-k/2}\sum_{x\in\mathcal A}\ket x$.
\end{proof}

\begin{corollary}[Non-local magic as locally injected magic]\label{cor:localtoNL}
Let $\bm p$ be a probability vector of length $2^m$ and
$\ket{\phi_{\bm p}}\coloneqq\sum_{i=0}^{2^m-1}\sqrt{p^\downarrow_{i}}\ket{i}$. Then
\begin{align}
D_{\mathrm{stab}}(\ket{\phi_{\bm p}})=D^{\mathrm{NL}}_{\mathrm{stab}}(\bm p).
\end{align}
Moreover, let $C$ be any Clifford on $S\otimes B$, $\ket s$ any stabilizer state of
$B$, and $\ket\phi$ any state of $S$ such that $C\ket\phi\ket s$ has Schmidt
coefficients $\sqrt{p_i}$ across $S\otimes B$. Then
$D_{\mathrm{stab}}(\ket\phi)\ge D^{\mathrm{NL}}_{\mathrm{stab}}(\bm p)$, with equality
for $\ket\phi=\ket{\phi_{\bm p}}$ and $C=CNOT$.
\end{corollary}
\begin{proof}
The first $2^k$ computational-basis labels form a linear subspace, so
\cref{lem:nonnegfid} gives
$F_{\mathrm{stab}}(\phi_{\bm p})\ge2^{-k}(\sum_{i\le2^k}\sqrt{p^\downarrow_i})^2$ for
every $k$, i.e.\ $D_{\mathrm{stab}}(\phi_{\bm p})\le D^{\mathrm{NL}}_{\mathrm{stab}}(\bm p)$.
Conversely, for $C,\ket s,\ket\phi$ as in the statement,
$D^{\mathrm{NL}}_{\mathrm{stab}}(\bm p)=D^{\mathrm{NL}}_{\mathrm{stab}}(C\ket\phi\ket s)
\le D_{\mathrm{stab}}(C\ket\phi\ket s)=D_{\mathrm{stab}}(\ket\phi\ket s)\le
D_{\mathrm{stab}}(\ket\phi)$. The last step uses
$F_{\mathrm{stab}}(\phi\otimes s)\ge F_{\mathrm{stab}}(\phi)$. For
$\ket\phi=\ket{\phi_{\bm p}}$, $\ket s=\ket0$ and the transversal CNOT, the state
$C\ket\phi\ket s=\sum_i\sqrt{p_i^\downarrow}\ket i\ket i$ has the required
Schmidt coefficients, so the two bounds meet.
\end{proof}

\cref{cor:localtoNL} gives $D^{\mathrm{NL}}_{\mathrm{stab}}$ the operational meaning as the minimal amount of magic that must be injected \emph{locally} so that a
Clifford interaction can produce a given entanglement spectrum. Equivalently, every
state $\rho$ is the marginal of $(U\otimes\id)\,C\,\ket{\phi_{\bm p}}\ket0$ with
$\bm p=\mathrm{spec}(\rho)$, and no Clifford-dilated preparation from a less magic
local input exists. Together with \cref{thm:orbitmin}, the minimal magic of
purification along the unitary orbit of $\rho$ is thus the minimal magic that must be
injected locally to prepare a purification of $\rho$ through a Clifford interaction. Applied to measurements, the non-local magic created between
system and apparatus equals $D^{\mathrm{NL}}_{\mathrm{stab}}$ of the Born
distribution. It vanishes if and only if the distribution is flat on $2^k$
outcomes, consistent with \cref{cor:gradedfilters}.

\subsection{Graded steering}\label{subsec:steering}

We now apply the Clifford dilation of projective measurements to quantum steering. We will see that magic is again the relevant resource to make steering graded. In this scenario, 
Alice and Bob share a pure state $\ket\psi_{AB}$ with Schmidt coefficients
$\theta_i$, $A$ of $n_A$ qubits and $B$ of $n_B$ qubits. Alice performs the
rank-one projective measurement $\{\Pi^\sigma_k=\dyad{\sigma_k}\}_k$, realized as
in \cref{prop:rankonedilation} by $W^\sigma=(\tilde U\otimes\id_M)\,C\,(\tilde U^\dagger\otimes\id_M)$
with a record register $M$ (a copy of $A$ prepared in $\ket0$) and
$\tilde U\ket k=\ket{\sigma_k}$. We call $\tilde U$ Alice's \emph{dressing}: all the
magic of her measurement is carried by $\tilde U$, while the interaction with the
record is the Clifford $C$. Outcome $k$ occurs with probability $p_k$ and leaves Bob
in the pure state $\ket{\beta_k}$, where
\begin{align}
\sqrt{p_k}\,\ket{\beta_k}\coloneqq\ket{b_k}\coloneqq(\bra{\sigma_k}\otimes\id_B)\ket\psi .
\label{eq:steeredvectors}
\end{align}
The collection $\{p_k,\beta_k\}$ is the \emph{assemblage} steered by the measurement.
By no-signalling, $\sum_kp_k\,\dyad{\beta_k}=\rho_B$ for every choice of basis.

\begin{lemma}[]\label{lem:steergram}
The Gram matrix of the steered vectors is the reduced state of Alice written in her
measurement basis:
\begin{align}
\langle b_k|b_l\rangle=\bra{\sigma_l}\rho_A\ket{\sigma_k}
=\bra{l}\tilde U^\dagger\rho_A\tilde U\ket{k}.
\label{eq:steergram}
\end{align}
In particular $p_k=\bra{\sigma_k}\rho_A\ket{\sigma_k}$. If $d_A=d_B=d$ and
$\ket\psi=(K\otimes\id)\ket\Omega$ as in \cref{sec:causality}, then
$p_k\dyad{\beta_k}=\tfrac1d\,\overline{K^\dagger\Pi^\sigma_kK}$. Measuring in the
dressed basis on $(K\otimes\id)\ket\Omega$ is equivalent to measuring in the
computational basis on $(\tilde U^\dagger K\otimes\id)\ket\Omega$.
\end{lemma}
\begin{proof}
$\langle b_k|b_l\rangle=\bra\psi(\ket{\sigma_k}\!\bra{\sigma_l}\otimes\id)\ket\psi
=\tr(\rho_A\ket{\sigma_k}\!\bra{\sigma_l})$. For the second statement,
$\tr_A[(\Pi\otimes\id)(K\otimes\id)\dyad\Omega(K^\dagger\otimes\id)]
=\tr_A[(K^\dagger\Pi K\otimes\id)\dyad\Omega]=\tfrac1d(K^\dagger\Pi K)^T$, and
$X^T=\overline X$ for Hermitian $X$. Finally,
$(\tilde U\otimes\id)^\dagger$ acting on $(K\otimes\id)\ket\Omega$ gives
$(\tilde U^\dagger K\otimes\id)\ket\Omega$, and
$(\tilde U^\dagger K)^\dagger(\tilde U^\dagger K)=K^\dagger K$.
\end{proof}

To classify assemblages, merge the outcomes
that leave Bob in the same state: this gives distinct states $\beta_c$ with weights
$w_c=\sum_{k\in c}p_k$.

\begin{definition}[Sharp and graded steering]\label{def:sharpsteering}
The assemblage of a measurement is \emph{sharp} if its distinct states $\beta_c$ are
pairwise orthogonal and all carry the same weight $w_c=2^{-j}$ for some
$j\in\mathbb N$. Otherwise it is \emph{graded}.
\end{definition}

In a sharp assemblage Bob learns about Alice's outcome in an all-or-nothing way while a graded assemblage contains either non-orthogonal steered states, which give Bob partial information, or non-dyadic weights. Stabilizer resources only produce sharp steering, as a consequence of the following lemma.

\begin{lemma}[]\label{lem:stabmeas}
Let $\ket\Phi$ be a pure stabilizer state on $X\otimes Y\otimes R$ with stabilizer
group $\mathrm{Stab}(\Phi)$. Measure on $X$ the joint eigenbasis of a maximal abelian
group $\mathcal G$ of Pauli operators on $X$ with $-\id\notin\mathcal G$, i.e.\ a
stabilizer basis. Label the outcomes by the characters $s$ of $\mathcal G$, with
projectors $\Pi_s=|\mathcal G|^{-1}\sum_{g\in\mathcal G}\chi_s(g)\,g$, and let
$\omega_s$ be the conditional state of $Y$. Then:
\begin{itemize}
\item[(a)] for any two outcomes $s,s'$ of non-zero probability, there is
$h=h_X\otimes h_Y\otimes h_R\in\mathrm{Stab}(\Phi)$ with
$(\Pi_{s'}\otimes\id)\ket\Phi=h\,(\Pi_s\otimes\id)\ket\Phi$; in particular
$p_{s'}=p_s$ and $\omega_{s'}=h_Y\,\omega_s\,h_Y^\dagger$;
\item[(b)] the conditional states are pairwise equal or orthogonal, and each distinct
state occurs with the same total probability $2^{-j}$.
\end{itemize}
\end{lemma}
\begin{proof}
(a) Let $\mathcal G_0\coloneqq\{g\in\mathcal G:\pm g\in\mathrm{Stab}(\Phi)\}$, and for
$g\in\mathcal G_0$ let $\lambda_g=\pm1$ be such that $\lambda_gg\in\mathrm{Stab}(\Phi)$.
Since $\Pi_sg=\chi_s(g)\Pi_s$, we have
$\Pi_s\ket\Phi=\lambda_g\Pi_sg\ket\Phi=\lambda_g\chi_s(g)\Pi_s\ket\Phi$. So every
outcome of non-zero probability satisfies $\chi_s=\lambda$ on $\mathcal G_0$, and for
two such outcomes the character $\mu=\chi_s\chi_{s'}$ is trivial on $\mathcal G_0$. Commutation with elements of $\mathrm{Stab}(\Phi)$ defines a linear map from
$\mathrm{Stab}(\Phi)$ to the characters of $\mathcal G$. Its image is the annihilator
of $\{g\in\mathcal G: g \text{ commutes with }\mathrm{Stab}(\Phi)\}$, which equals
$\mathcal G_0$ because a maximal stabilizer group coincides, up to signs, with its own Pauli commutant. Hence there is $h\in\mathrm{Stab}(\Phi)$ with
$hgh^\dagger=\mu(g)\,g$ for all $g\in\mathcal G$, i.e.\ $h\Pi_sh^\dagger=\Pi_{s'}$.
Then $\Pi_{s'}\ket\Phi=h\Pi_sh^\dagger\ket\Phi=h\Pi_s\ket\Phi$. Tracing out $XR$
and using cyclicity gives $\omega_{s'}=h_Y\omega_sh_Y^\dagger$ and $p_{s'}=p_s$.

(b) Each $\omega_s$ is a marginal of the stabilizer state
$\Pi_s\ket\Phi/\sqrt{p_s}$, so $\omega_s\in\mathrm{STAB0}$ is a normalized projector
onto a stabilizer code. Conjugating it by a Pauli yields the projector onto a syndrome
space of the same stabilizer group, which is either the same space or an orthogonal
one. The characters $\mu$ trivial on $\mathcal G_0$ form a group $\Gamma$ that acts
simply transitively on the outcomes of non-zero probability, which are therefore
$|\Gamma|$ in number and equiprobable. By (a), the subset
$\Gamma_0=\{\mu\in\Gamma:\omega_{\mu s}=\omega_s\}$ does not depend on $s$ (Pauli
conjugations commute up to signs) and is a subgroup. The classes of equal
conditional states are the cosets of $\Gamma_0$, so they have equal weight
$|\Gamma_0|/|\Gamma|=2^{-j}$.
\end{proof}

Thanks to the above result, we can now easily prove that stabilizer steering is sharp:

\begin{proposition}[]\label{prop:stabsteering}
If $\ket\psi_{AB}$ is a stabilizer state and Alice measures in a stabilizer basis
(Clifford dressing $\tilde U$), the assemblage is sharp.
\end{proposition}
\begin{proof}
Apply \cref{lem:stabmeas} with $X=A$, $Y=B$ and $R$ trivial. The conditional states
of $B$ are pure, since Alice's measurement is rank one.
\end{proof}

Graded steering therefore requires magic, either in the shared state or in Alice's
dressing. The next two results separate the two sources and connect sharp steering with the absence of non-local magic.

\begin{proposition}[Sharp steering and flatness]\label{prop:sharpflat}
Alice has a measurement with a sharp assemblage if and only if
$D^{\mathrm{NL}}_{\mathrm{stab}}(\ket\psi)=0$, i.e.\ the Schmidt spectrum is flat on
$2^j$ values. In particular, if the entanglement spectrum is not flat, \emph{every}
measurement of Alice, with any dressing, steers Bob in a graded way.
\end{proposition}
\begin{proof}
If an assemblage is sharp, no-signalling gives
$\rho_B=\sum_c2^{-j}\dyad{\beta_c}$ with orthonormal $\beta_c$, i.e.\
$\rho_B=\Pi/2^j$. The Schmidt spectrum is then flat on $2^j$ values, and
\cref{th:mainth} gives $D^{\mathrm{NL}}_{\mathrm{stab}}=0$. Conversely,
$D^{\mathrm{NL}}_{\mathrm{stab}}(\ket\psi)=0$ means
$2^{-k}(\sum_{i\le2^k}\theta^\downarrow_i)^2=1$ for some $k$. By Cauchy--Schwarz,
$2^{-k}(\sum_{i\le2^k}\theta^\downarrow_i)^2\le\sum_{i\le2^k}(\theta_i^\downarrow)^2\le1$,
with equality iff the spectrum is flat on $2^k$ values. Measuring in a Schmidt basis
then steers Bob into his $2^k$ orthonormal Schmidt vectors with probability $2^{-k}$
each, which is sharp.
\end{proof}

\begin{proposition}\label{prop:allsharp}
Every rank-one projective measurement of Alice produces a sharp assemblage if and
only if $\ket\psi$ is a product state or $\rho_A=\id/d_A$.
\end{proposition}
\begin{proof}
If $\rho_A$ is pure, all $\ket{b_k}$ are parallel by \cref{eq:steergram}, which gives
a single class of weight $1$. If $\rho_A=\id/d_A$, \cref{eq:steergram} gives
$\langle b_k|b_l\rangle=\delta_{kl}/d_A$: $d_A=2^{n_A}$ orthonormal, equiprobable
states for every basis. In this case Alice's dressing only changes Bob's frame, since
$\ket\psi=(\id\otimes V)\ket{\Omega_{AA'}}$ with $V$ an isometry from a copy $A'$ of
$A$ into $B$, and therefore $(\tilde U\otimes\id)\ket\psi=(\id\otimes V\tilde U^TV^\dagger)\ket\psi$,
where $V\tilde U^TV^\dagger$ extends to a unitary on $B$.

Conversely, suppose $r=\mathrm{rank}\,\rho_A$ satisfies $2\le r$. If $r=d_A$ and
$\rho_A\neq\id/d_A$, \cref{prop:sharpflat} excludes any sharp measurement. If
$2\le r<d_A$, pick eigenvectors $u_1,u_2$ of $\rho_A$ with eigenvalues
$\lambda_1,\lambda_2>0$ and a unit vector $w\in\ker\rho_A$, and complete
$\sigma_1=(u_1+w)/\sqrt2$ and $\sigma_2=(u_1-w)/2+u_2/\sqrt2$ to an orthonormal basis.
By \cref{eq:steergram} the corresponding Bob states satisfy
$|\langle\beta_1|\beta_2\rangle|^2=\lambda_1/(\lambda_1+2\lambda_2)\in(0,1)$, so the
assemblage is graded.
\end{proof}

Maximally entangled states therefore absorb any magic in Alice's dressing into Bob's
frame. For flat but non full rank spectra the filter of \cref{lem:choi} is a sharp
projector, yet a dressing that is misaligned with its support compresses Alice's
basis vectors non-orthogonally and produces graded steering. In that case the
gradedness is supplied entirely by Alice's local magic. For instance, take
$\ket\psi=\ket{\Phi^+}_{A_1B}\ket0_{A_2}$ and the non-stabilizer basis
$\ket{\sigma_1}=\ket0\ket+$,
$\ket{\sigma_{2,3}}=\tfrac12(\ket{00}-\ket{01})\pm\tfrac1{\sqrt2}\ket{10}$,
$\ket{\sigma_4}=\ket{11}$. It yields $\bm p=(\tfrac14,\tfrac38,\tfrac38,0)$ and Bob
states with $|\langle\beta_1|\beta_{2}\rangle|^2=\tfrac13$ and
$|\langle\beta_2|\beta_{3}\rangle|^2=\tfrac19$, although
$D^{\mathrm{NL}}_{\mathrm{stab}}(\ket\psi)=0$.

To quantify the magic in a steering scenario, consider the joint state of Alice's
system, the record and Bob after the dilated measurement,
\begin{align}
\ket\Psi_{AMB}\coloneqq(W^\sigma\otimes\id_B)\ket\psi_{AB}\ket0_M
=\sum_k\sqrt{p_k}\,\ket{\sigma_k}_A\ket k_M\ket{\beta_k}_B ,
\label{eq:steeringstate}
\end{align}
and define the \emph{steering magic}
\begin{align}
\mathfrak T(\psi,\sigma)\coloneqq\min_{V_A,U_B}
D_{\mathrm{stab}}\big((V_A\otimes\id_M\otimes U_B)\ket\Psi\big).
\end{align}
Here the record basis is kept fixed, while Alice's post-measurement frame and Bob's
frame, which do not affect the assemblage, are optimized. The next theorem gives quantitative results for the amount of magic involved in steering.

\begin{theorem}[]\label{thm:steeringmagic}
For every $\ket\psi_{AB}$ and every basis $\{\ket{\sigma_k}\}$:
\begin{itemize}
\item[(i)] $\max\big\{D^{\mathrm{NL}}_{\mathrm{stab}}(\ket\psi),\,D^{\mathrm{NL}}_{\mathrm{stab}}(\bm p)\big\}
\le\mathfrak T(\psi,\sigma)\le\min_{U_B}D_{\mathrm{stab}}\big((\tilde U^\dagger\otimes U_B)\ket\psi\big)$;
\item[(ii)] if $\mathfrak T(\psi,\sigma)=0$, the assemblage is sharp;
\item[(iii)] $\min_\sigma\mathfrak T(\psi,\sigma)=D^{\mathrm{NL}}_{\mathrm{stab}}(\ket\psi)$,
attained by measuring in a Schmidt basis with the Schmidt coefficients in decreasing
order along the computational order of the record.
\end{itemize}
\end{theorem}
\begin{proof}
(i) For any cut $X|Y$ and any state, $D_{\mathrm{stab}}\ge D^{\mathrm{NL}}_{\mathrm{stab}}$
across $X|Y$, and $V_A\otimes\id_M\otimes U_B$ is local with respect to both
$AM|B$ and $A|MB$. Across $AM|B$ the Schmidt coefficients of $\ket\Psi$ are those of
$\ket\psi$, because $W^\sigma$ acts unitarily on $AM$ and $M$ starts in $\ket0$.
Across $A|MB$ they are $\sqrt{p_k}$ by \cref{eq:steeringstate}. The lower bound
follows from \cref{th:mainth}, since non-local magic depends only on the Schmidt
spectrum. For the upper bound choose $V_A=\tilde U^\dagger$:
$(\tilde U^\dagger\otimes\id_M\otimes U_B)\ket\Psi=C\,\big((\tilde U^\dagger\otimes U_B)\ket\psi\big)\ket0_M$,
whose log-stabilizer fidelity is at most $D_{\mathrm{stab}}((\tilde U^\dagger\otimes U_B)\ket\psi)$
by Clifford invariance and $F_{\mathrm{stab}}(\phi\otimes0)\ge F_{\mathrm{stab}}(\phi)$.

(ii) The minimum is attained, since $F_{\mathrm{stab}}$ is a maximum of finitely many
continuous functions and the unitary groups are compact. If it vanishes, then
$(V_A\otimes\id\otimes U_B)\ket\Psi=\sum_k\sqrt{p_k}\,V_A\ket{\sigma_k}\ket k\,U_B\ket{\beta_k}$
is a stabilizer state, and reading the record is a computational-basis measurement
of $M$. Apply \cref{lem:stabmeas} with $X=M$, $Y=B$, $R=A$. The conditional states
$U_B\ket{\beta_k}$ are pure, so the assemblage of Bob, which $U_B$ does not affect,
is sharp.

(iii) By (i), $\mathfrak T\ge D^{\mathrm{NL}}_{\mathrm{stab}}(\ket\psi)$ for every
basis. If $\tilde U$ maps $\ket i$ to the $i$-th Schmidt vector of $A$ (decreasing
$\theta_i$), a suitable $U_B$ gives
$(\tilde U^\dagger\otimes U_B)\ket\psi=\sum_i\theta^\downarrow_i\ket i\ket i$. By
\cref{th:mainth} its log-stabilizer fidelity equals
$D^{\mathrm{NL}}_{\mathrm{stab}}(\ket\psi)$, so the upper bound in (i) meets the lower
bound.
\end{proof}

\cref{thm:steeringmagic} identifies non-local magic as the irreducible magic of any
steering scenario built on $\ket\psi$. Whatever Alice's dressing adds on top of it is
magic in the relative frame between her measurement and the Schmidt basis. For maximally entangled states it is useless (\cref{prop:allsharp}) and is the only source
of gradedness for flat spectra. The measurement interaction itself is always the
Clifford $C$ and never contributes.

As an example, let $\ket\psi=\sqrt{0.9}\ket{00}+\sqrt{0.1}\ket{11}$. Measuring
$Z$ gives orthogonal Bob states with non-dyadic probabilities $(0.9,0.1)$. Measuring
$X$ gives uniform probabilities, but Bob states $\sqrt{0.9}\ket0\pm\sqrt{0.1}\ket1$
with overlap $0.64$. In both cases $\mathfrak T=D^{\mathrm{NL}}_{\mathrm{stab}}(\ket\psi)=\log(10/9)$
(for $X$, the upper bound in (i) with $U_B=\id$ reduces to
$D_{\mathrm{stab}}(\ket\psi)$ because the Hadamard is Clifford). 

\ignore{
\begin{remark}\label{rem:steeringmagic}
$\mathfrak T$ is a magic budget for the whole assemblage, not a steering monotone. It
also counts gradedness that stays on Alice's side. For a product state with non-flat
$\bm p$, $\mathfrak T\ge D^{\mathrm{NL}}_{\mathrm{stab}}(\bm p)>0$ although there is no
steering: this is the non-local magic created between Alice's system and the record,
as in \cref{cor:localtoNL}. Isolating the part of $\mathfrak T$ that is genuinely due
to the correlations with Bob is left for future work.
\end{remark}
}

\section{Conclusions}

We have provided an exact characterization of non-local magic for the log-stabilizer fidelity and connected it to the magic of purification of the reduced state. A natural next step is to determine whether analogous closed-form expressions in terms of the entanglement spectrum can be obtained for other magic monotones, in particular the stabilizer R\'enyi entropies (SREs), and whether the minimizers of SRE-based non-local magic coincide with those identified here for the stabilizer fidelity. 

A natural perspective would be extending the non-local log-stabilizer fidelity to multipartite settings, where it is expected to be sensitive to the structure of multipartite entanglement and, in principle, could be adopted to signal the scale at which magic resided in many-body systems. Naturally, monotones of non-local magic in multipartite systems cannot rely on closed-form expressions based on the Schmidt spectrum of the state, and performing the numerical optimization over the subsystems' unitary groups is unfeasible already at small system sizes. Therefore, finding computable measures or bounds for multipartite non-local magic remains an open challenge, and a key step toward a unified picture of the quantum resources underlying complex quantum states and quantum correlations.

Beyond states, \cref{sec:dilation} shows that non-local magic also measures the magic that must be injected locally to dilate quantum operations through Clifford interactions. 
The same framework shows that graded steering requires magic, either non-local magic of the shared state, which is irreducible (\cref{thm:steeringmagic}), or local magic in the dressing of the steering measurement.
These results suggest several directions. Firstly, steering scenarios could be turned into witnesses that certify non-local magic without full knowledge of the devices. Secondly, the dilation picture points toward a notion of non-local magic for quantum channels. It would be interesting to further explore how such a quantity relates to adaptiveness, and whether it can quantify the resources provided by classical control in stabilizer protocols.
More broadly, our results identify non-local magic as an operationally meaningful resource, setting the cost of tasks in which entanglement alone is not sufficient.
%; isolating a genuine steering monotone within this magic budget (\cref{rem:steeringmagic}) is an open problem. It would also be interesting to further explore the connection between non-local magic, causality, and quantum steering, with the goal of clarifying whether non-stabilizerness can provide a quantitative resource for more general forms of non-local influence and directional quantum correlations.

\medskip

{\bf Note added.} During the completion of this manuscript, a closely related study obtained the same closed-form characterization of the stabilizer-fidelity optimization in \cref{th:mainth} and developed complementary consequences for entanglement embezzlement~\cite{Sierant2026}.

\medskip

{\bf Funding Declaration.} No funding to declare.

\medskip

{\bf Data Availability.} No data is available.

\medskip

{\bf Author Contribution.} All authors contributed equally to this work.

\medskip

{\bf  Corresponding author.} Michele Viscardi, mviscardi@unisa.it

\medskip

{\bf Competing Interest declaration.} No Competing Interests to declare.

\bibliographystyle{apsrev4-1}
\bibliography{LVH}

\let\oldaddcontentsline\addcontentsline% Store \addcontentsline
\renewcommand{\addcontentsline}[3]{}% Make \addcontentsline a no-op

\let\addcontentsline\oldaddcontentsline
\appendix
\onecolumngrid
\clearpage
\begin{center}
    {\normalfont\Large\bfseries Appendix}
\end{center}
\setcounter{secnumdepth}{2}
\setcounter{equation}{0}
\setcounter{figure}{0}
\setcounter{table}{0}
\setcounter{section}{0}
\renewcommand{\thetable}{S\arabic{table}}

\renewcommand{\thefigure}{S\arabic{figure}}
\renewcommand{\thesection}{S.\arabic{section}}
\counterwithout{equation}{section}
\renewcommand{\theequation}{S\arabic{equation}}

\section{Lemmata}\label{app:nlm}
\begin{lemma}\label{lem1}
    Let $\ket{\psi},\ket{\phi}$ two bipartite states with Schmidt spectrum $\lambda_i^2,\theta_j^2$. Then
    \begin{align}
        |\langle\psi|\phi\rangle|\le \sum_i{\lambda_i\theta_i}
    \end{align}
    with equality iff they are have the same Schmidt basis (ordered).
    \begin{proof}
        We can write every bipartite state as
        \begin{align}
            \ket{\psi}=\sum_{p,q}X_{pq}\ket{p}\ket{q}
        \end{align}
        where $\ket{p}$ (resp. $\ket{q}$) is a basis for $A$ (resp. $B$). Notoriously, the singular values of $X$ correspond to the Schmidt spectrum of $\ket{\psi}$. Hence, we have $X=U\lambda V^{\dagger}$ and $Y=R\theta W^{\dagger}$ with $Y$ the coefficient matrix corresponding to $\ket{\phi}$. We can write the overlap
        \begin{align}
            |\langle\psi|\phi\rangle|=|\tr(Y^{\dagger}X)|
        \end{align}
        and we need to show that $|\tr(Y^{\dagger}X)|\le \sum_{i}\lambda_i\theta_i$. First notice that
        \begin{align}
            \tr(Y^{\dagger}X)=\tr(\theta R^{\dagger}U\lambda V^{\dagger}W)
        \end{align}
        Defining $A=R^{\dagger}U$ and $B=V^{\dagger}W$, we have
        \begin{align}
            |\tr(Y^{\dagger}X)|=|\tr(\theta A\lambda B)|=\Big|\sum_{i,j}\theta_iA_{ij}\lambda_jB_{ji}\Big|\le \sum_{i,j}\theta_i\lambda_j |A_{ij}||B_{ji}|
        \end{align}
        Defining $C_{ij}=|A_{ij}||B_{ji}|$, let us show that $\sum_{i}C_{ij},\sum_{j}C_{ij}\le 1$. We have
        \begin{align}
            \sum_{j}C_{ij}=\sum_{j}|A_{ij}||B_{ji}|\le \Big(\sum_{j}|A_{ij}|^2\Big)^{1/2}\Big(\sum_{j}|B_{ji}|^2\Big)^{1/2}\le 1
        \end{align}
        where we used Cauchy-Schwartz inequality and the fact that both $A$ and $B$ are unitaries so:
        \begin{align}
            \sum_{j}|A_{ij}|^2=\sum_{j}A_{ij}A_{ij}^{*}=\sum_jA_{ij}(A^{\dagger})_{ji}=1
        \end{align}

    Now, we define $a_i=\lambda_i-\lambda_{i+1}$ and $b_i=\theta_i-\theta_{i+1}$ by identifying $\theta_{2^n+1}=0$. We can clearly write
    \begin{align}
        \lambda_i=\sum_{j\ge i}a_{j},\quad \theta_i=\sum_{j\ge i}b_j
    \end{align}
    implying
    \begin{align}
        \sum_{ij}\theta_i\lambda_j C_{ij}=\sum_{ij}\sum_{k\ge j}\sum_{l\ge i}C_{ij}a_{k}b_l=\sum_{i,j,k,l}C_{ij}a_kb_l\delta_{j\le k}\delta_{i\le l}=\sum_{kl}a_kb_l\sum_{i\le l,j\le k}C_{ij}
    \end{align}
    Now, notice that we can write
    \begin{align}
        \sum_{i\le l,j\le k}C_{ij}\le \sum_{i\le l}\sum_jC_{ij}\le l+1
    \end{align}
    and
    \begin{align}
        \sum_{i\le l,j\le k}C_{ij}\le \sum_{j\le k}\sum_iC_{ij}\le k+1
    \end{align}
    meaning that
    \begin{align}
        \sum_{ij}\theta_i\lambda_j C_{ij}\le \sum_{kl}a_kb_l\min\{k+1,l+1\}
    \end{align}
    On the other hand, we also have
    \begin{align}
        \sum_{i}\theta_i\lambda_i=\sum_{i,k,l}\delta_{i\le k}\delta_{i\le l}a_{k}b_l=\sum_{kl}a_kb_l\sum_{i}\delta_{i\le k}\delta_{i\le l}=\sum_{k,l}a_kb_l\min\{k+1,l+1\}
    \end{align}
    which proves the statement.
        \end{proof}
\end{lemma}

\begin{lemma}[Compression bound]\label{lem:compression}
Let $A\ge 0$ have eigenvalues $\lambda_1^\downarrow\ge\lambda_2^\downarrow\ge\dots$
and let $\Pi$ be a projector of rank $s$. Then
\begin{align}
\norm{A\Pi}_1\ \le\ \sum_{i=1}^{s}\lambda_i^\downarrow ,
\end{align}
with equality if and only if $\operatorname{supp}\Pi$ is a direct sum of
eigenspaces of $A$ exhausting its $s$ largest eigenvalues (with multiplicity).
\begin{proof}
$\norm{A\Pi}_1=\tr\sqrt{\Pi A^2\Pi}$. Let $\mu_1\ge\dots\ge\mu_s$ be the
eigenvalues of the compression $\Pi A^2\Pi|_{\operatorname{supp}\Pi}$. By Cauchy's
interlacing theorem $\mu_j\le(\lambda_j^\downarrow)^2$, hence
$\norm{A\Pi}_1=\sum_j\sqrt{\mu_j}\le\sum_{j\le s}\lambda_j^\downarrow$.
Equality forces $\mu_j=(\lambda_j^\downarrow)^2$ for every $j$, whence
$\tr(A^2\Pi)=\sum_{j\le s}(\lambda_j^\downarrow)^2$. This is the Ky Fan maximum of
$\tr(A^2\Pi)$ over rank-$s$ projectors, attained precisely when
$\operatorname{supp}\Pi$ is spanned by eigenvectors of $A^2$ belonging to its $s$
largest eigenvalues. Conversely, for such a $\Pi$ the compression is diagonal with
$\mu_j=(\lambda_j^\downarrow)^2$ and equality holds.
\end{proof}
\end{lemma}

\end{document}